\documentclass[twoside,11pt]{article}
\usepackage{jair, theapa, rawfonts}

\usepackage{balance}
\usepackage{amsthm}
\usepackage{comment}
\usepackage{graphicx}

\usepackage{booktabs}
\usepackage{tabularx}
\usepackage{caption}
\usepackage{url}

\usepackage{algorithm}
\usepackage{algorithmic}
\usepackage{enumitem}

\newtheorem{example}{Example}
\newtheorem{theorem}{Theorem}

\newtheorem{definition}{Definition}
\newtheorem{corollary}[theorem]{Corollary}
\newtheorem{observation}[theorem]{Observation}

\newtheorem{proposition}[theorem]{Proposition}

\jairheading{1}{1993}{1-15}{6/91}{9/91}
\ShortHeadings{Online Housing Market}
{Julien Lesca}
\firstpageno{1}

\begin{document}

\title{Online Housing Market}

\author{\name Julien Lesca \email julien.lesca@dauphine.fr \\
       \addr Universit\'e Paris-Dauphine, Universit\'e PSL, CNRS, LAMSADE, 75016, Paris, France}

% For research notes, remove the comment character in the line below.
% \researchnote

\maketitle

\begin{abstract}
This paper studies an online variant of the celebrated housing market problem \cite{SHAPLEY197423}, where each agent owns a single house and seeks to exchange it for another based on her preferences. In this online setting, agents may arrive and depart at any time, meaning that not all agents are present on the housing market simultaneously. We extend the well-known serial dictatorship and Gale’s top trading cycle mechanisms to this online scenario, aiming to retain their desirable properties such as Pareto efficiency, individual rationality, and strategy-proofness. These extensions also seek to prevent agents from strategically delaying their arrival or advancing their departure. We demonstrate that achieving all of these properties simultaneously is impossible in the online context, and we present several variants that achieve different subsets of these properties.\end{abstract}

\section{Introduction}

Allocating indivisible resources to agents is a fundamental problem in computational social choice, which lies at the intersection of economics \cite{THOMSON2011393} and computer science \cite{Klaus_Manlove_Rossi_2016,manlove2013algorithmics}. The decision maker in such problems must consider both the preferences of the agents over the resources and their strategic behavior. In this paper, we focus on a specific problem known as the housing market in the matching theory literature \cite{SHAPLEY197423}, where each agent is endowed with a single resource that they are willing to exchange for another. Furthermore, we assume that no monetary compensation is allowed to offset any unfavorable exchanges during the process. Despite its simplicity, this problem has many applications, including the exchange of dormitory rooms for students \cite{RePEc:aea:aecrev:v:92:y:2002:i:5:p:1669-1686}, the trade of used items \cite{swapz}, kidney exchanges where incompatible donor/recipient pairs exchange viable organs for transplants \cite{ferrari2015kidney}, and many more (see, e.g., \cite{BiroTRENDS2017} for additional applications).

In this paper, we assume that the preferences over resources are ordinal, which is a standard assumption for the housing market. By doing so, we implicitly assume that the decision maker cannot assess the magnitude of preferences between resources, either because this information is unavailable or because it is too costly to obtain through interaction with the agents. The procedure for reallocating resources among agents is centralized, and the decision maker aims to produce an allocation that is as efficient as possible

Under ordinal preferences, Pareto optimality is an appropriate efficiency measure, aiming to find allocations where no agent can be made better off without harming another. In the standard offline setting, this requirement can be satisfied by the serial dictatorship (SD) procedure, sometimes called picking-sequence \cite{BouveretL14}. In this procedure, agents are arranged in a specific order, and each agent, one by one, selects their most preferred resource from the pool of remaining resources. However, this process is not individually rational, as some agents may end up with resources less desirable than their initial endowments. This is undesirable because agents may be incentivized to avoid participating in the trade for fear of being worse off. The well-known Gale's top trading cycle (TTC) procedure \cite{SHAPLEY197423} circumvents this drawback by ensuring individual rationality while computing an allocation that is Pareto optimal.

Since the procedure is centralized, the decision maker must interact with agents to gather their preferences over resources. However, during this process, agents may misreport their preferences in an attempt to manipulate the procedure and achieve a more favorable outcome. Mechanism design, a subfield of game theory, seeks to create procedures where truthfully revealing preferences is a dominant strategy for agents \cite{Hurwicz_Reiter_2006}. Both serial dictatorship and TTC procedures possess this property. Moreover, it has been proven that the TTC procedure is the only mechanism that is simultaneously individually rational, Pareto efficient, and strategy-proof \cite{ma1994strategy}.

The standard offline setting assumes that all agents participating in the exchange are available at the same time, and that the exchange procedure is conducted during this period. However, in many contexts, this requirement is unrealistic or too demanding, making it difficult to gather a larger set of agents for the exchange. In more realistic scenarios, agents are only available during restricted time periods, meaning that some may not be able to participate in the exchange simultaneously. This is the case, for example, with online exchange websites \cite{swapz}, where agents do not arrive at the marketplace at the same time and cannot remain indefinitely before their swaps take place. This type of scenario is referred to as online \cite{albers2003online}, or dynamic, in the context of matching \cite{baccara2021dynamic}. In this paper, we show how the standard procedures, serial dictatorship and TTC, can be adapted to fit an online environment.

\subsection*{Outline}

Section~\ref{secRelatedWork} provides a non-exhaustive list of related works connected either to the housing market problem or to online versions of social choice problems. Section~\ref{secPrem} presents the main definitions of the key properties that we seek to achieve for our mechanisms. Sections~\ref{secStatSerDict} to~\ref{secSafeSerDict} describe online mechanisms based on the famous serial dictatorship procedure. More specifically, Section~\ref{secStatSerDict} first considers a static variant of this procedure, where assignments to non-leaving agents are irrevocable. Section~\ref{secDynSerDict} relaxes this requirement by considering a dynamic version of the serial dictatorship. Section~\ref{secSafeSerDict} considers a safe version in which no agent can be allocated an item she prefers less than her own. Finally, Section~\ref{secTTClike} introduces multiple online variants of the TTC procedure.

\section{Related Works}\label{secRelatedWork}

Related works encompass studies on variants of the housing market problem and investigations into online versions of social choice problems.

\subsection{Housing Market Problem Variants}

Multiple variants of the standard housing market have been explored in the literature. First, the extension to exchanges involving multiple items has been extensively studied. Fujita et al. \cite{fujita2018complexity} proposed a TTC-like procedure that selects a Pareto optimal allocation (which belongs to the core) under restricted preferences over sets of resources. This work was extended to problems where multiple copies of items exist, and similar results were found under slightly different restrictions on preferences \cite{SikdarAX19}. When agents can exchange multiple resources, the problem of computing an individually rational and Pareto efficient allocation is known to be NP-hard, even for additive preferences \cite{DBLP:conf/atal/AzizBLLM16}. A sufficient condition to achieve strategy-proofness has also been provided \cite{Aziz20}.

Other extensions have considered exchange situations where agents are embedded in a social network. In some cases, the possible exchanges are constrained by the structure of the social network \cite{damamme2015power,ijcai2017p31,saffidine2018constrained,huang2019object}. Others have viewed the social network as a means to advertise the exchange procedure, designing mechanisms to incentivize agents to invite their neighbors to participate \cite{kawasaki2021mechanism,you2022strategy}.

Finally, other works have explored extensions of the preference model used to describe agents' preferences. Standard algorithms have been adapted to handle cases with indifferences \cite{DBLP:conf/aaai/AzizK12,saban2013house}. In another direction, the scope of the preference model can extend beyond simply considering the resource allocated to an agent. For instance, some works \cite{lesca2018service,DBLP:conf/atal/0001L20} have introduced preference models that also take into account the identity of the agent receiving an individual's initial endowment.

\subsection{Online Problems in Social Choice}

Multiple social problems related to our exchange problem have been examined through the lens of online procedures. In fair division, which aims to allocate resources fairly, several online variants have been explored \cite{DBLP:conf/aaai/AleksandrovW20,SankarLNN21,hosseini2024class}. For instance, strategy-proofness and Pareto efficiency have been studied in this context \cite{DBLP:conf/ijcai/AleksandrovAGW15}, alongside envy-freeness. Various extensions of the serial dictatorship procedure have been proposed to achieve these properties \cite{aleksandrov2019strategy}. The online electric vehicle charging problem \cite{ijcai2019p773}, which focuses on scheduling charging for customers arriving in an online fashion, also shares similarities with our problem. However, in both cases, the online setting differs from mine, as resources do not arrive dynamically and are known to the procedure from the outset.  

The online setting has also been explored in the context of stable matching, where preferences are two-sided and pairs of agents are matched to achieve stability, with applications such as student-university allocations. The case where students or universities arrive in an online fashion has been studied \cite{doval2022dynamically}, and mechanisms based on the standard deferred acceptance algorithm have been proposed to address this issue. Online updates to the instance may also arise from the fact that the matching process consists of multiple rounds, during which agents may unexpectedly or strategically alter their revealed preferences between rounds \cite{ijcai2019p19,RePEc:spr:joecth:v:68:y:2019:i:2:d:10.1007_s00199-018-1133-9,DBLP:conf/atal/BampisEY23}. In the same vein, but closer to our setting as it addresses the one-sided problem of assigning resources to agents, an online version where preferences are elicited incrementally by querying agents has been considered \cite{hosseini2021necessarily}. A procedure that elicits preferences to achieve Pareto-optimality while minimizing the number of queries has been proposed.

Finally, numerous works have focused on dynamic kidney exchange \cite{10.1111/j.1467-937X.2009.00575.x,doi:10.1287/mnsc.2020.3954}, where donors and/or recipients arrive in an online fashion. Most of these studies consider compatibility (0-1 preference models) rather than ordinal preferences. Furthermore, they often assume the allocation process can be probabilistic, rely on probabilistic assumptions about arrival and/or departure, and aim to reduce the expected waiting time before a transplant \cite{10.1257/mic.20150183,RePEc:inm:oropre:v:65:y:2017:i:6:p:1446-1459,https://doi.org/10.3982/TE3740}. Others focus on maximizing the expected number of matched pairs \cite{awasthi2009online,dickerson2012dynamic}. The design of strategy-proof and individually rational mechanisms has also been considered for transplant centers participating in national exchange programs \cite{RePEc:eee:gamebe:v:91:y:2015:i:c:p:284-296,DBLP:conf/aaai/HajajDHSS15}. 

It is worth noting that there is a substantial body of literature on online algorithms \cite{borodin2005online}, which aim to maximize the sum of the weights of matched pairs in an online assignment problem where weights are assigned to pairs to be matched \cite{mehta2013online}. However, in our context, these approaches assume that preferences are cardinal and aim to maximize the agents' social welfare, i.e., the sum of their utilities. Their goal is to design online algorithms with a high competitive ratio, where the competitive ratio is the smallest ratio between the value of the computed allocation and the value of the optimal offline allocation.

\section{Preliminary}\label{secPrem}

Let $ N = \{ 1, \ldots, n \} $ denote the set of agents. Each agent $ i $ owns a single good\footnote{The resources in this paper are interchangeably referred to as goods or items} $ e_i $, called her initial endowment. Let $ \mathcal{T} = [t^-, t^+] $ denote the timeline during which the market is open, where $ t^- $ and $ t^+ $ represent the opening and closing times of the market, respectively. For each agent $ i $, $ a_i $ and $ d_i $ represent her arrival and departure times in the market, with $ a_i, d_i \in \mathcal{T} $ such that $ t^- \leq a_i < d_i \leq t^+ $. To simplify the setting, we assume that no two arrival or departure times occur simultaneously.\footnote{Note that when the exchange market takes place in a web-based application, times are recorded precisely, and equal timestamps are highly unlikely. Furthermore, if such an event occurs, random tie-breaking rules may be used.} Let $ E = \{ e_1, \ldots, e_n \} $ represent the set of items belonging to the agents in $ N $. Each agent $ i $ has a strict ordinal preference $ \succ_i $ over the items in $ E $, such that $ e_j \succ_i e_k $ means agent $ i $ strictly prefers $ e_j $ over $ e_k $. Furthermore, symbol $\succeq_i$ stands for $\succ_i$ or $=$. An instance of the online exchange problem is represented by the set of tuples $ \{ (i, e_i, a_i, d_i, \succ_i) \}_{i \in N} $. Without loss of generality, we assume that the agent indices are ordered by increasing arrival times, i.e., such that $a_1 < a_2 < \ldots < a_n$ hold. Let $ \mathcal{I} $ denote the entire set of possible instances.

Our goal is to allocate to each agent a single item such that each item is assigned only once. In other words, we search for an allocation $ M: N\rightarrow E $ such that $ M(i)\neq M(j) $ for each $ i\neq j $. The timeline constraints are such that each agent will leave the market with an item which belongs to an agent that arrived in the market earlier than her departure time. For any $ t\in \mathcal{T} $, let $ N_{<t} $ denote the subset of agents arriving before $ t $, i.e., containing any agent $ i\in N $ such that $ a_i<t $. For a given instance $ I\in \mathcal{I} $, an allocation $ M $ is $ I $-compatible if $ M(i)\in E_{<d_i} $ holds for each agent $ i $, where $ E_{<t}=\{ e_i : i\in N_{<t}\} $ for each $ t\in \mathcal{T} $. For any instance $I$, the set of $ I $-compatible allocations is denoted $ \mathcal{M}(I) $.

An exchange algorithm $ \mathcal{A} $ returns an $ I $-compatible allocation for each instance $ I $. By abuse of notation, we assume that $ \mathcal{A}_i(I) $ denotes the item allocated to agent $ i $ by algorithm $ \mathcal{A} $ for instance $ I $. In the online version of the problem, the algorithm should make a decision on the allocation for an agent when she leaves the market, without knowing the agents that arrive after her departure. In other words, the mechanism should make a decision only based on the agents that already visited the market by the departure time of the agent. For any $ t \in \mathcal{T} $, let $ I_{<t} $ denote a truncated copy of $ I $ restricted to the agents of $ N_{<t} $.
\begin{definition}
An exchange algorithm $ \mathcal{A} $ is online if for each instance $ I $ and for each agent $ i $, $ \mathcal{A}_i(I) = \mathcal{A}_i(I_{<d_i}) $ holds.
\end{definition}

The choice between online exchange algorithms should be guided by the properties fulfilled by the allocation that they compute. To compare the online exchange algorithms, we consider multiple standard desiderata properties that could be considered as desirable or necessary. The first one is a standard definition of efficiency in multi-agent decision problems.

\begin{definition}\label{defPO}
Allocation $ M' $ Pareto-dominates allocation $ M $ if for each agent $ i $, $ M'(i) \succeq_i M(i) $ holds, and for at least one agent $ j $, $ M'(j) \succ_j M(j) $ holds. For a given instance $ I $, let $ \mathcal{S} $ denote a given subset of $ \mathcal{M}(I) $. Allocation $ M $ is $ \mathcal{S} $-Pareto optimal ($ \mathcal{S} $-$ PO $), if there is no allocation $ M' $ of $ \mathcal{S} $ that Pareto-dominates it.
\end{definition}

The standard notion of Pareto optimality corresponds to $ \mathcal{M}(I) $-$ PO $, as illustrated by the following example.

\begin{example}\label{exSD}
    Consider instance $I$ described in Figure \ref{figOHMexample}. The timeline, containing the arrival and departure times of the agents, as well as their preferences are provided in the left part of the figure. Three different allocations are graphically described in the right part of the figure. For example, allocation $M$ is such that $M(1) = e_2$, $M(2) = e_3$, and $M(3) = e_1$. Note that $M$ is not $I$-compatible, as agent 2 receives an item from an agent who arrives later than $d_2$. Therefore, $M$ does not belong to $\mathcal{M}(I)$. On the other hand, allocations $M'$ and $M''$ are both $I$-compatible. It is also easy to verify that both $M'$ and $M''$ are $\mathcal{M}(I)$-$PO$.

    \begin{figure}[t]
\includegraphics[width=\textwidth]{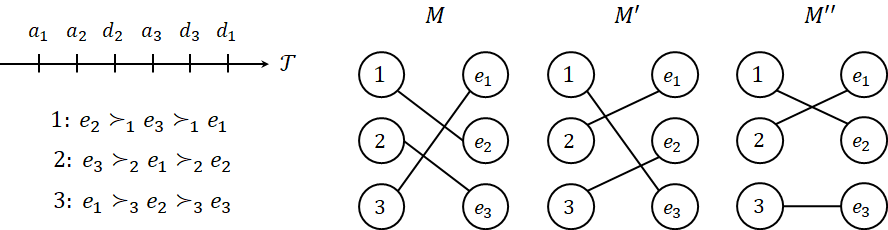}
\caption{Online housing market example with three agents.}\label{figOHMexample}
\end{figure}
\end{example}

As suggested in Example \ref{exSD}, restricting comparisons to $I$-compatible allocations allows for the existence of efficient solutions, even if better ones exist but are incompatible with the online setting.

\begin{definition}
An exchange algorithm $ \mathcal{A} $ is individually rational ($ IR $), if for each instance $ I $, $ \mathcal{A}_i(I) \succeq_i e_i $ holds for any agent $ i $.
\end{definition}

The individual rationality property belongs to the set of properties that incentivize agents to participate in the exchange algorithm, as no agent will receive an item that is less desirable than her initial endowment. Another standard property related to the manipulative power of agents, through misreporting their preferences, is called incentive compatibility. We generalize the notion of incentive compatibility to also take into account manipulations related to arrival and departure times.

\begin{definition}\label{defIncentiveCompatible}
An exchange algorithm is strongly incentive compatible ($ SIC $), if for each agent $ i $ and for each pair of instances $ I $ and $ I' $ such that $ I' = I \setminus \{i, e_i, a_i, d_i, \succ_i\} \cup \{i, e_i, a'_i, d'_i, \succ'_i\} $ and $ a_i \leq a'_i < d'_i \leq d_i $, $ \mathcal{A}_i(I) \succeq_i \mathcal{A}_i(I') $ holds. It is $ a $-$ IC $ (respectively $ d $-$ IC $) if the above inequality holds only when $ d_i = d'_i $ (respectively $ a_i = a'_i $).\footnote{Note that $ a $-$ IC $ is close to the notion of participation that appears in \cite{DBLP:conf/ijcai/MatteiSW17}, except that $a$-$IC$ requires also incentive compatibility on revealed preferences.} Finally, it is weakly incentive compatible ($ WIC $), if the above inequality holds only when both $ a_i = a'_i $ and $ d_i = d'_i $.
\end{definition}

In Definition~\ref{defIncentiveCompatible}, instance~$I$ refers to the setting in which all agents truthfully report their preferences and arrival/departure times, while $I'$ is identical except that agent~$i$ may misreport her preferences and/or availability. We aim for mechanisms in which such misreports do not yield a more favorable outcome for agent~$i$. The notion denoted $WIC$ in our setting, where agent~$i$ cannot misreport her arrival or departure time, corresponds to the standard notion of incentive compatibility in the offline setting. The arrival time $a_i$ represents the earliest possible time at which agent~$i$ can enter the market—she cannot arrive earlier due to external constraints (e.g., unavailability, continued use of the item). Similarly, $d_i$ denotes the latest possible departure time. This is why agent~$i$ cannot arrive before $a_i$ or stay beyond $d_i$.\footnote{We implicitly assume it is possible to verify an agent’s presence in the market during her declared interval. For instance, arrival can be confirmed by collecting the item she offers, and continued presence until $d_i$ can be enforced by delivering her assigned item only at that time.} However, agent~$i$ may still strategically choose to arrive later than $a_i$ or leave earlier than $d_i$ in an attempt to manipulate the outcome. Our goal is to design mechanisms that incentivize agents to remain in the market throughout their true availability window, from $a_i$ to $d_i$. The different notions of incentive compatibility aim to ensure that agents are incentivized to arrive as early as possible ($a$-IC), as late as possible ($d$-IC), or both (SIC).

\section{Static Serial Dictatorship Procedure}\label{secStatSerDict}

In this section, we consider an exchange algorithm based on the serial dictatorship procedure, which is described in Algorithm \ref{algoPickingSequence}. The standard version of the serial dictatorship procedure is based on a permutation of the agents defining the order in which each agent will choose her item among the remaining ones. We slightly generalize these permutations by considering permutation functions whose order depends on the instance. More formally, $\Pi: \mathcal{I} \rightarrow N^N$ denotes this permutation function, which is such that for any instance $I$ and for any $i\in \{1, \ldots , n\}$, $\Pi_i(I)$ denotes the agent choosing at position $i$. Despite the fact that this permutation theoretically may depend on the full instance, we assume in this paper that it only depends on the arrival and departure times of the agents. An example of such a permutation is the ascending departure permutation $\delta$ which ranks the agents according to their departure times, and more specifically by increasing departure times, i.e., such that $d_{\delta(1)} < d_{\delta(2)} < \ldots < d_{\delta(n)}$.\footnote{Note that, to simplify notation, we omit the input instance in the notation of the concrete permutation functions ($\delta$ and $\alpha$, as well as the permutation $\pi$ appearing in Algorithms~\ref{algoPickingSequence}, \ref{algoOnlinePickingSequence}, and \ref{algoIROnlinePickingSequence}). Moreover, an agent's position in these permutations is indicated in parentheses rather than as a subscript.}

\begin{algorithm}[t]
\caption{Static online serial dictatorship procedure}
\textbf{Input}:  Permutation function $\Pi$.

%\textbf{Output}:  $\mathcal{A}$-optimal matching $M$.

\begin{algorithmic}[1]\label{algoPickingSequence}
\STATE Initialize $M$ as an empty matching.
\STATE $A\leftarrow \emptyset$. \COMMENT{Items already assigned.}
\STATE $B\leftarrow \emptyset$ \COMMENT{Agents already matched.}
\STATE $j\leftarrow 1$ \COMMENT{Position in $\Pi$ of the first remaining agent to choose.}
 \FOR{$i=1$ \textbf{to} $n$}
 \STATE \COMMENT{Iteration occurring at $d_{\delta(i)}$.}
   \IF{$\delta(i)\not\in B$} 
     \STATE Let $\pi$ denotes $\Pi(I_{< d_{\delta(i)}})$.
     \WHILE{$\pi(j)\neq \delta(i)$}
       \STATE $M(\pi(j))\leftarrow best_{\pi(j)}(E_{< d_{\delta(i)}}\setminus A)$.
      \STATE $A \leftarrow A\cup \{ M(\pi(j))\}$.
       \STATE $B\leftarrow B\cup \{ \pi(j)\}$
       \STATE $j\leftarrow j + 1$.
     \ENDWHILE
     \STATE $M(\delta(i))\leftarrow best_{\delta(i)}(E_{< d_{\delta(i)}}\setminus A)$.
     \STATE $A \leftarrow A\cup \{ M(\delta(i))\}$.
     \STATE $B\leftarrow B\cup \{ \delta(i)\}$
     \STATE $j\leftarrow j + 1$.
   \ENDIF
 \ENDFOR
\STATE \textbf{return} $M$.
\end{algorithmic}
\end{algorithm}

In Algorithm \ref{algoPickingSequence}, every iteration of the "for" loop corresponds to the departure of an agent, specifically agent $\delta(i)$, whose departure time is the $i^{th}$ earliest. However, the order in which the agents choose their items is determined by the permutation function $\Pi$. Therefore, before agent $\delta(i)$ chooses her item, all agents that are ranked before her according to $\Pi$ must choose their items. In this static (or greedy) version of the algorithm, once an agent has chosen an item, she is permanently matched to it, even if more desirable items arrive later. Note that the procedure $best_i$ returns the most preferred item of agent $i$ from a given set of items. Algorithm \ref{algoPickingSequence} is designed to be online, as only the part of the instance corresponding to agents who have already arrived, $I_{< d_{\delta(i)}}$, is used at iteration $i$, during the departure of agent $\delta(i)$. However, the permutation must be prefix-consistent i.e., the order of agents already matched to items must not change when new agents arrive. Otherwise, an agent might be matched multiple times or ignored by the algorithm.

\begin{definition}\label{defPS}

Permutation function $\Pi$ is prefix-stable ($PS$), if for each agent $i$, and for any positions $i'$ and $j'$ such that $\Pi_{i'}(I)=i$ and $j'\leq i'$, $\Pi_{j'}(I)=\Pi_{j'}(I_{< d_i})$ holds.

\end{definition}

This definition stipulates that, for a permutation function to be prefix-stable, any agents ranked before agent~$i$ when she leaves the market should retain their positions afterward. The following proposition shows that Algorithm~\ref{algoPickingSequence} is online and weakly incentive compatible when the permutation used is prefix-stable.

\begin{proposition}\label{propOnlineWIC}
    Algorithm \ref{algoPickingSequence} is both online and $WIC$ if the input permutation function is $PS$.
\end{proposition}

\begin{proof}
Let us first show that it is online. Let $I$ be an instance and $i$ be an agent. We need to show that Algorithm \ref{algoPickingSequence} applied to both $I$ and $I_{< d_i}$ returns the same outcome to agent $i$. To simplify notations, we denote by $I'$ instance $I_{< d_i}$. Let $i'$ denote the position of agent $i$ in $\Pi(I)$, i.e., $\Pi_{i'}(I)=i$ holds. Note first that since $\Pi$ is $PS$, we know that for any position $j'$ such that $j'\leq i'$, we have $\Pi_{j'}(I)=\Pi_{j'}(I')$. For the same reason, we know that for any agent $j$ and positions $j'$ and $k'$ such that $k'\leq j'\leq i'$ and $\Pi_{j'}(I)=\Pi_{j'}(I')=j$, we have $\Pi_{k'}(I)=\Pi_{k'}(I_{< d_{j}})$ and $\Pi_{k'}(I')=\Pi_{k'}(I'_{< d_{j}})$. But since $I'=I_{< d_i}$, this implies $I_{< d_{j}}=I'_{< d_{j}}$ and $\Pi_{k'}(I_{< d_{j}})=\Pi_{k'}(I'_{< d_{j}})$ hold for any $k'\leq j'\leq i'$. Therefore, the first iterations of the ``for'' loop of Algorithm \ref{algoPickingSequence}, until agent $i$ leaves the market at time $d_i$, makes the same assignment to the same agents for $I$ and $I'$. Therefore, agent $i$ receives the exact same item in both cases.

We now show that it is $WIC$. The proof relies on the fact that the standard serial dictatorship procedure is incentive compatible in the offline setting. This is due to the fact that the preferences of an agent are only considered once to assign to this agent her most favorite item among the remaining ones (by function $best$). Therefore, an agent has no incentive to misreport her preferences since the item that she will receive cannot be more preferred by doing so.
\end{proof}

\subsection{The Ascending Departure Permutation}

As suggested by the following proposition, the ascending departure permutation $\delta$, that orders the agents by increasing departure time, is a good candidate to be used with Algorithm \ref{algoPickingSequence}.

\begin{proposition}\label{propDeltaOnline}
The ascending departure permutation $\delta$ is $PS$.
\end{proposition}

\begin{proof}
We must show that for each agent $i$, and for any positions $i'$ and $j'$, such that $\Pi_{i'}(I)=i$ and $j'\leq i'$, the equality $\Pi_{j'}(I)=\Pi_{j'}(I_{< d_i})$ holds when $\Pi=\delta$. Note first that agents $\Pi_1(I')$, $\Pi_2(I')$, $\ldots$, $\Pi_{i'}(I')$ are the $i'$ agents with the smallest departure times in $I'$. Instance $I$ is an upper set of $I'$. Furthermore, any agent $k$ that belongs to $I$ but not to $I'$ must have an arrival time later than $d_i$, since $d_k > a_k > d_i$. Therefore, the $i'$ agents with the smallest departure times are the same in both $I$ and $I'$. Finally, their relative order remains consistent in both $\Pi(I)$ and $\Pi(I')$ because their departure times are identical in both instances.
\end{proof}

\begin{example}
    Let us run Algorithm \ref{algoPickingSequence} with the ascending departure permutation $ \delta $ on the instance described in Figure \ref{figOHMexample}. At the departure of agent 2, she selects item $ e_1 $ that she prefers to $ e_2 $. Then, at the departure of agent 3, she selects item $ e_2 $ that she prefers to $ e_3 $. Finally, agent 1 leaves the market with $ e_3 $, which is the only remaining item. The resulting allocation corresponds to $ M' $ described in Figure \ref{figOHMexample}.
\end{example}

Note that when Algorithm \ref{algoPickingSequence} is used with the ascending departure permutation $\delta$, the algorithm allows each agent to choose the best remaining item upon leaving the market. The following proposition demonstrates that selecting the ascending departure permutation $\delta$ as the permutation function in Algorithm \ref{algoPickingSequence} results in a Pareto efficient outcome for any instance.

\begin{proposition}\label{propParetoOpt} For any instance $I\in \mathcal{I}$, Algorithm \ref{algoPickingSequence}, using the ascending departure permutation $\delta$ as input, returns a $\mathcal{M}(I)$-$PO$ allocation. \end{proposition}

\begin{proof}
The main argument of the proof is that the matching $M$ returned by Algorithm \ref{algoPickingSequence} must be lexicographically optimal among the matchings of $\mathcal{M}(I)$ according to the order provided by the ascending departure permutation $\delta$ applied to $I$. Let $\mathcal{P}$ denote the subset of allocations that may Pareto-dominate $M$. According to Algorithm \ref{algoPickingSequence}, $M(\delta(1))$ must be the most favored item of agent $\delta(1)$ in $E_{< d_{\delta(1)}}$. Therefore, all matchings in $\mathcal{P}$ must allocate $M(\delta(1))$ to agent $\delta(1)$ to Pareto-dominate $M$. Provided that $M(\delta(1))$ is allocated to agent $\delta(1)$, Algorithm \ref{algoPickingSequence} allocates to agent $\delta(2)$ her most favored item in $E_{< d_{\delta(2)}} \setminus \{M(\delta(1))\}$. Thus, all matchings in $\mathcal{P}$ must allocate $M(\delta(2))$ to agent $\delta(2)$ to Pareto-dominate $M$. By applying the same reasoning to the remaining agents in the order given by $\delta$, we conclude that $\mathcal{P}$ can only contain a single allocation, which corresponds to $M$. Since a matching cannot Pareto-dominate itself, no matching in $\mathcal{M}(I)$ Pareto-dominates $M$.
\end{proof}

The following proposition shows that there is no other online exchange algorithm that always returns a Pareto-efficient outcome.

\begin{proposition}\label{propOnlyPO}
Algorithm \ref{algoPickingSequence} using the ascending departure permutation $\delta$ is the only online exchange algorithm that returns for any instance $I$ an $\mathcal{M}(I)$-$PO$ allocation.
\end{proposition}

\begin{proof}
    By contradiction, let $\mathcal{A}$ denote an online exchange algorithm that behaves differently than Algorithm \ref{algoPickingSequence} using the ascending departure permutation $\delta$ as input. This means that there exists an instance $I$ and an agent $j$ such that $\mathcal{A}_j(I)$ does not correspond to the assignment made by Algorithm \ref{algoPickingSequence} with the ascending departure permutation $\delta$. In other words, agent $j$ does not receive her most favorite item among those in $E_{< d_{\delta(j)}}$ minus the items already assigned to agents leaving earlier than her. Furthermore, we assume without loss of generality that, among the set of agents in this situation, agent $j$ is the one who leaves the earliest. Let $M$ denote the allocation obtained by Algorithm \ref{algoPickingSequence} with the ascending departure permutation $\delta$ applied to $I$. Let $i$ denote the agent who receives $M(j)$ instead of agent $j$ in $\mathcal{A}(I)$, i.e., such that $\mathcal{A}_i(I)=M(j)$. By assumption, $d_j < d_i$ must hold since, otherwise, an agent arriving earlier than agent $j$ would not receive the same item as in $M$.
    
    To show the contradiction, we construct a new instance $I'$ by adding dummy agents to $I$. For any agent $k$ leaving after $d_j$, let $k'$ be a dummy agent such that $a_{k'} = d_{j} + k\epsilon$, with $\epsilon > 0$ small enough for all these agents to arrive earlier than any other arrival or departure, and $d_{k'}$ large enough to leave after any other agent of $I$. The preferences of the agents in $I'$ are as follows. For any agent $k$ that was already in $I$ and leaving after $d_j$, her most favorite item is $e_{k'}$, followed by the items of $I$ ranked in the same order, and finally all the other items of $I'$ ranked in arbitrary order. Dummy agent $i'$'s preferences are $\mathcal{A}_j(I)\succ_{i'} M(j)\succ_{i'} \ldots$, where all the other items are ranked in arbitrary order. Finally, the most favorite item of any other dummy agent $k'$ is $\mathcal{A}_k(I)$, and all other items are ranked arbitrarily.
    
    Note first that since $\mathcal{A}$ is online, both $\mathcal{A}_j(I')=\mathcal{A}_j(I'_{< d_j})$ and $\mathcal{A}_j(I_{< d_j})=\mathcal{A}_j(I)$ hold. Furthermore, since all dummy agents arrive later than $d_j$, $I_{< d_j}=I'_{< d_j}$ holds, which implies with the former equalities that $\mathcal{A}_j(I')=\mathcal{A}_j(I)$. Now concerning the assignment of the other agents, based on their preferences and since $\mathcal{A}$ is $\mathcal{M}(I')$-$PO$, we can say that for each dummy agent $k'$, we have $\mathcal{A}_{k'}(I')=\mathcal{A}_k(I)$, and for each regular agent $k$ leaving after $d_j$, we have $\mathcal{A}_k(I')=e_{k'}$. More precisely for $i'$ we have $\mathcal{A}_{i'}(I')=\mathcal{A}_i(I)=M(j)$. Note that except for agents $i'$ and $j$, any agent receives her most favorite item.
    
    But now consider the almost same allocation, say $M'$, as the one constructed by $\mathcal{A}$, except that agent $j$ receives $M(j)$ and agent $i'$ receives $\mathcal{A}_j(I)$. Note that by construction, $M'$ should belong to $\mathcal{M}(I')$ since $\mathcal{A}$ is an online exchange algorithm, and therefore $\mathcal{A}(I_{<d_j})=\mathcal{A}(I'_{<d_j})$ belongs to $\mathcal{M}(I'_{<d_j})$ (for agent $j$ and the ones leaving after her), $\mathcal{A}(I')$ belongs to $\mathcal{M}(I')$ (for all other agents except agent $i'$), and $a_{i'}<d_i$ holds (for agent $i'$). Agent $j$ and dummy agent $i'$ receive in $M'$ a strictly better item than the one offered by algorithm $\mathcal{A}$. Therefore, $M'$ Pareto dominates $\mathcal{A}(I')$, leading to a contradiction since $\mathcal{A}$ is $\mathcal{M}(I')$-$PO$.

\end{proof}

\subsection{The Ascending Arrival Permutation}\label{secAscendingArrival}

Propositions \ref{propParetoOpt} and \ref{propOnlyPO} tell us that when we want an online algorithm that always returns an efficient outcome, we should focus on Algorithm \ref{algoPickingSequence} using the ascending departure permutation $\delta$ as input. Furthermore, Proposition \ref{propOnlineWIC} and \ref{propDeltaOnline} attest that this algorithm is $WIC$. However, it is easy to check that Algorithm \ref{algoPickingSequence} using the ascending departure permutation $\delta$ as input is neither $a$-$IC$, $d$-$IC$ nor $IR$.\footnote{\label{footnotePOnotIR}Proposition~\ref{alphaOnlyDIC} shows that the only permutation function that makes Algorithm~\ref{algoPickingSequence} $d$-IC differs from $\delta$. Furthermore, Proposition~\ref{propSerDictIncomAIC} states that no permutation function makes Algorithm~\ref{algoPickingSequence} $a$-IC. Finally, Example~\ref{exSafeSerialDictator} presents an instance where the matching returned by Algorithm~\ref{algoPickingSequence} using $\delta$ as input is not $IR$.} Therefore, we should consider other permutation functions as input of Algorithm \ref{algoPickingSequence}. Another candidate is the ascending arrival permutation $\alpha$, which ranks the agents by increasing arrival time i.e., such that $a_{\alpha(1)}<a_{\alpha(2)}<\ldots < a_{\alpha(n)}$. Note that this permutation corresponds to the identity because we have assumed that $a_1<a_2<\ldots < a_n$. The following proposition shows that  the ascending arrival permutation $\alpha$ is prefix-stable.

\begin{proposition}\label{propAlphaOnline}
The ascending arrival permutation $\alpha$ is $PS$.
\end{proposition}
    
\begin{proof}
Note first that for any instance $I$ and for any agent $i$ and $j$ such that $a_j < a_i$, $a_i < d_i$ implies that $a_j < d_i$, and therefore agent $j$ belongs to $I_{< d_i}$. Furthermore, for any agent $k$ that belongs to $I \setminus I_{<d_i}$, $a_i < d_i < a_k$ hold, and therefore $k$ cannot be ranked before agent $i$ according to the ascending arrival permutation $\alpha$ applied to $I$. Finally, the arrival time of each of these agents is the same in $I$ and $I_{< d_i}$. Therefore, the positions of the $i$ first ranked agents according to  the ascending arrival permutation $\alpha$ are the same for $I$ and $I_{< d_i}$.
\end{proof}

\begin{example}
    Let us run Algorithm \ref{algoPickingSequence} with  the ascending arrival permutation $ \alpha $ on the instance described in Figure \ref{figOHMexample}. At the departure of agent 2, agent 1 selects first and picks $ e_2 $, which she prefers to $ e_1 $. Then, agent 2 selects the last remaining item at $ d_2 $, which is $ e_1 $. Finally, at the departure of agent 3, she picks the last remaining item at $ d_3 $, which is $ e_3 $. The resulting allocation corresponds to $ M'' $ described in Example \ref{exSD}.
\end{example}

The following proposition shows that the choice of the ascending arrival permutation $\alpha$ as input of Algorithm \ref{algoPickingSequence} incentivizes agents to truthfully report their departure time.

\begin{proposition}\label{propAlphaDIC}
Algorithm \ref{algoPickingSequence} using  the ascending arrival permutation $\alpha$ as input is $d$-$IC$.
\end{proposition}

\begin{proof}
Note first that the ascending arrival permutation $\alpha$ is $PS$ according to Proposition \ref{propAlphaOnline}, and therefore by Proposition \ref{propOnlineWIC}, we know that no agent has an incentive to misreport her preferences. So we only need to show that no agent has an incentive to misreport her departure time. Note also that as soon as the arrival time is not modified, changing the departure time of an agent has no impact on the ranking dictated by the ascending arrival permutation $\alpha$. We are going to show that if the reported departure time of an agent, say $i$, is delayed after a new event, which may be the departure of another agent or the arrival of a new agent, then her most favorite remaining item is either the same or even better. More formally, let $I$ be an instance and $I'$ be the same instance except that the new departure time $d'_i$ of agent $i$ is delayed to be later than the arrival time $a_j$ or the departure time $d_j$ corresponding to the event occurring right after $d_i$ i.e., either \textit{(i)} $d_i<d_j<d'_i$ or \textit{(ii)} $d_i<a_j<d'_i$. In both cases \textit{(i)} and \textit{(ii)}, all iterations of the ``for'' loop of Algorithm \ref{algoPickingSequence} applied to instance $I$ and $I'$ before reaching $d_i$ are the same, since only the departure time of agent $i$ has changed. Note that if agent $i$ has to choose during one of these iterations for $I$, agent $i$ chooses the same item for $I'$ and is ultimately assigned the same item. Therefore, we can focus on the case where agent $i$ does not choose until it leaves the market at $d_i$ in $I$.

Let us first consider case \textit{(i)}, where agent $j$ leaves earlier than $i$ in $I'$ whereas she left later in $I$. We consider at the same time the two iterations of the ``for'' loop corresponding to $d_i$ and $d_j$ in $I$, and $d_j$ and $d'_i$ in $I'$, which are the first two iterations that differ between $I$ and $I'$. Note that since there is no event occurring between $d_i$ and $d_j$ in $I$, and between $d_j$ and $d'_i$ in $I'$, at the end of these two iterations, the same subset of agents picks an item for $I$ and $I'$, which should be a subset of agents ranked before either agent $i$ or $j$ in the ascending arrival permutation $\alpha$. Furthermore, they choose their items according to the same order (the one dictated by the ascending arrival permutation $\alpha$), and from the same subset of items. Therefore, each one of them chooses the exact same item for both $I$ and $I'$, including agent $i$ who has no incentive to delay her arrival time.

Finally, let us consider case \textit{(ii)}, where $d_i<a_j<d'_i$. Note that a new item arrives to the market, namely $e_j$, before agent $i$ leaves the market. This new item may have consequences on the choices made by the agents who choose before $i$ according to the ascending arrival permutation $\alpha$. More specifically, one of them may choose $e_j$ instead of the item they selected for $I$. Once item $e_j$ has been chosen, say by agent $k$, it is no longer available, but the item that agent $k$ chose in $I$ at the same phase is now available and takes the place of $e_j$ as an additional item to choose from. So each agent choosing before $i$ either picks the same item as in $I$, or picks $e_j$, or another item abandoned by an agent choosing earlier. Therefore, when it is the turn of agent $i$ to choose, the item assigned to her in $I$ is available, as well as an additional item, which may be $e_j$ or another item abandoned by an agent choosing earlier. Consequently, she picks either the same item as in $I$ or an item that she prefers. 
 
 In summary, if agent $i$ delays her departure time, then she receives either the same item or a new one that she prefers. In other words, she has no incentive to report an earlier departure time than her true one, and the algorithm is $d$-$IC$.

\end{proof}

The following proposition shows that no other permutation function, apart from the ascending arrival permutation $\alpha$, makes Algorithm \ref{algoPickingSequence} $d$-$IC$.

\begin{proposition}\label{alphaOnlyDIC}
The ascending arrival permutation $\alpha$ is the only input that ensures Algorithm \ref{algoPickingSequence} is both online and $d$-$IC$.
\end{proposition}

\begin{proof}
    Let $\Pi$ be an online permutation function that is different from the ascending arrival permutation $\alpha$. This means that there is an instance $I$ with two agents $i$ and $j$, whose positions according to $\Pi(I)$ are $i'$ and $j'$, respectively, such that both $a_i > a_j$ and $i' < j'$ hold. In other words, agent $i$ arrives later but is ranked before agent $j$ according to $\Pi$. By contradiction, assume that $\Pi$ makes Algorithm \ref{algoPickingSequence} $d$-$IC$. Let $I'$ be a copy of instance $I$ with the exception of the departure time $d'_j$ of agent $j$, which happens earlier than $d_j$. More precisely, we assume that $a_j < d'_j < a_i$ hold. Since Algorithm \ref{algoPickingSequence} applied with $\Pi$ is $d$-$IC$, the item assigned to agent $j$ for $I'$ should be no better than the one assigned to her for $I$. Furthermore, since Algorithm \ref{algoPickingSequence} is online, agent $j$ should receive the exact same item for instances $I'_{< d'_j}$ and $I'$. Assume that the preferences of the agents are such that $e_j$ is the most favorite item of agents $i$ and $j$, and any other agent $k$ prefers her own item $e_{k}$ to any other one. It is clear that in that case, each agent receives her own item during Algorithm \ref{algoPickingSequence} applied to $I'_{< d'_j}$ since each agent prefers her own item (note that agent $i$ is not part of $I'_{< d'_j}$). It implies that agent $j$ receives also $e_j$ when Algorithm \ref{algoPickingSequence} is applied to $I'$ with the ascending arrival permutation $\alpha$ as input. Since agent $j$ receives her most favorite item for $I'$, she should receive the exact same item in $I$. But since agent $i$ is ranked before agent $j$ in  the ascending arrival permutation $\alpha$, she picks first and chooses $e_j$, which is also her most favorite item. Therefore, agent $i$ should also receive item $e_j$ for $I$, leading to a contradiction since $e_j$ can be assigned only once.
\end{proof}

Note that this statement is not as strong as Proposition \ref{propOnlyPO} since it does not rule out any other $d$-$IC$ online exchange algorithm than Algorithm \ref{algoPickingSequence} using the ascending arrival permutation $\alpha$ as input. In the next sections, we will see another $d$-$IC$ algorithm based on the top trading cycle procedure.

With Algorithm \ref{algoPickingSequence} using  the ascending arrival permutation $\alpha$, we obtain the property of $d$-$IC$. It is easy to check that the property of $a$-$IC$ is not achieved with the ascending arrival permutation $\alpha$,\footnote{This is in fact a direct consequence of Proposition~\ref{propSerDictIncomAIC}.} and we wonder whether there is a permutation function that achieves this property with Algorithm \ref{algoPickingSequence}. The answer provided by the following proposition rules out the existence of such a permutation function.

\begin{proposition}\label{propSerDictIncomAIC}
There is no online permutation function that makes Algorithm \ref{algoPickingSequence} $a$-$IC$.
\end{proposition}

\begin{proof}
By contradiction, assume that there exists an online permutation $\Pi$ that makes Algorithm \ref{algoPickingSequence} $a$-$IC$. Consider the following instance $I$ with 3 agents, such that $a_1<a_2<d_2<a_3<d_3<d_1$. The preferences of the agents are $e_3\succ_1 e_1\succ_1 e_2$, $e_1\succ_2 e_2\succ_2 e_3$, and $e_1\succ_3 e_3\succ_3 e_2$. Note that since the algorithm is $a$-$IC$, agent 1 should be assigned item $e_3$, which is her most favorite item. Otherwise, by changing her arrival time to $a'_1$ such that $d_2<a'_1<a_3$, leading to a new instance $I'$, she would arrive later than agent 2's departure but before agent 3's departure. Whatever order is dictated by $\Pi(I')$, agent 1 will pick item $e_3$ and agent 3 will pick item $e_1$ (their most favorite items).

 Let us now consider instance $I_{< d_2}$, where agent 3 is suppressed. For this instance, if agent 2 is ranked first in $\Pi(I_{< d_2})$, then during the first iteration of the ``for'' loop of Algorithm \ref{algoPickingSequence}, she chooses item $e_1$, which is her most favorite item. During the next iteration, where agent 1 is alone in the market until she leaves at time $d_1$, she picks the only remaining item, $e_2$. However, since agent 1 prefers item $e_1$, she has an incentive to misreport her true arrival time and declare instead $a'_1$, such that $d_2<a'_1<d_1$, in order to retain item $e_1$. Therefore, since Algorithm \ref{algoPickingSequence} using $\Pi$ is $a$-$IC$, permutation function $\Pi(I_{< d_2})$ should be such that agent 1 chooses first. Hence, $\Pi_1(I_{< d_2})=1$ and $\Pi_2(I_{< d_2})=2$ should hold.

 But since permutation function $\Pi$ is online, we should have $\Pi_2(I) = \Pi_2(I_{< d_2}) = 2$ and $\Pi_1(I) = \Pi_1(I_{< d_2}) = 1$. Now, consider the application of Algorithm \ref{algoPickingSequence} with $\Pi$ to $I$. During the first iteration of the ``for'' loop, agent 1, who is the first to choose according to $\Pi(I)$, should choose an item among $e_1$ and $e_2$, leading to a contradiction with the fact that agent 1 is assigned item $e_3$.

\end{proof}

\section{Dynamic Serial Dictatorship Procedure}\label{secDynSerDict}

In this section, we explore a dynamic version of Algorithm \ref{algoPickingSequence}, where the choice made by the serial dictatorship procedure can be revised when new items arrive to the market. The final decision about the item allocated to an agent is made when this agent leaves the market. As described in Algorithm \ref{algoOnlinePickingSequence}, a serial dictatorship procedure is used to decide which item will be allocated to the agent leaving the market. Each agent appearing before the agent leaving the market in the sequence selects an item among the remaining ones, and after that, the agent leaving the market picks the best remaining item left by the others. Note that, except for the agent that leaves the market, no item is allocated to an agent. We can see the choices made by these agents remaining in the market as reservations that save all selected items to be picked by the leaving agent. However, the remaining agents may leave the market with other items, including those that will appear in the market later on. Note that Algorithm \ref{algoOnlinePickingSequence} makes use of the ascending departure permutation $\delta$ presented in the previous section, in order to treat the agents by increasing departure times, which is required for the algorithm to be online.

\begin{algorithm}[tb]
\caption{Dynamic serial dictatorship procedure}
\textbf{Input}:  Permutation function $\Pi$.

\begin{algorithmic}[1]\label{algoOnlinePickingSequence}
\STATE $A\leftarrow \emptyset$. \COMMENT{Items already assigned.}
\STATE $B\leftarrow \emptyset$ \COMMENT{Agents already matched.}
 \FOR{$i=1$ \textbf{to} $n$}
   \STATE \COMMENT{Iteration occurring at $d_{\delta(i)}$.}
   \STATE Let $\pi$ denote $\Pi(I_{< d_{\delta(i)}})$.
   \STATE $j \leftarrow 1$
   \STATE $T\leftarrow \emptyset$ \COMMENT{Set of items reserved in this phase.}
   \WHILE{$\pi(j) \neq \delta(i)$}
     \IF{$\pi(j)\not\in B$}
       \STATE $T \leftarrow T\cup \{ best_{\pi(j)}(E_{< d_{\delta(i)}}\setminus (A\cup T))\}$.
     \ENDIF
     \STATE $j\leftarrow j+1$
   \ENDWHILE
     \STATE $M(\delta(i))\leftarrow best_{\delta(i)}(E_{< d_{\delta(i)}}\setminus (A\cup T))$.
     \STATE $A \leftarrow A\cup \{ M(\delta(i))\}$.
     \STATE $B \leftarrow B\cup \{ \delta(i)\}$.
 \ENDFOR
\STATE \textbf{return} $M$.
\end{algorithmic}
\end{algorithm}

As with Algorithm \ref{algoPickingSequence}, Algorithm \ref{algoOnlinePickingSequence} should use an online permutation function as input to be online. However, contrary to Algorithm \ref{algoPickingSequence}, Algorithm \ref{algoOnlinePickingSequence} does not require the use of a $PS$ permutation function to run properly or to be online. However, incentive compatibility is more difficult to obtain since the procedure may ask multiple times an agent about her most favorite item among the remaining ones, and agents may have greater incentive to misreport. Therefore, Algorithm \ref{algoOnlinePickingSequence} requires the use of specific permutation functions to incentivize agents to truthfully report their preferences. Here is a definition describing this type of permutation functions.

\begin{definition}
    Permutation function $\Pi$ is prefix-stable extended ($PFE$), if for any three agents $i,j$ and $k$ such that $a_i<d_k<d_i$, $d_k<d_j$, $d_i>a_j$, agent $k$ is ranked behind agent $i$ in $\Pi(I_{< d_k})$, and either $a_j>d_k$ or $j$ is also ranked behind agent $i$ in $\Pi(I_{< d_k})$, then agent $j$ is ranked behind agent $i$ in $\Pi(I_{< \min{\{d_i,d_j\}}})$.
\end{definition}

To put it in simple terms, the order over the agents that are ranked before the leaving agent should not change afterward.\footnote{As the name suggests, the prefix-stable extended concept is closely related to the notion of prefix stability.} This means that an agent $i$ that is ranked in $\Pi(I_{<d_k})$ before the leaving agent $k$, should not be ranked later on after another agent $j$ that was either ranked after $i$ in $\Pi(I_{<d_k})$ or not even present before agent $k$ leaves the market. The intuition behind this formal definition is that the order in which agents reserve their items during Algorithm \ref{algoOnlinePickingSequence} should not change from iteration to iteration, because otherwise they may adopt a strategic behavior. For example, when agent $k$ is leaving the market, agent $i$ may strategically reserve an item for agent $j$ in order to avoid that later on when agent $j$ chooses before agent $i$, she picks or reserves an item that is better for agent $i$\footnote{The full explanation is provided in the proof of Proposition \ref{ICprop}}. The following proposition characterizes the $PFE$ permutation functions as the ones which make Algorithm \ref{algoOnlinePickingSequence} $WIC$.

\begin{comment}

\begin{definition}
A permutation function $\Pi$ is compatible if for any agents $i, j, k\in N$, such that $d_i>d_k$, $d_j>d_k$, $\Pi^{-1}_{I_{< d_k}}(j)>\Pi^{-1}_{I_{< d_k}}(i)$ and $\Pi^{-1}_{I_{< d_k}}(k)>\Pi^{-1}_{I_{< d_k}}(i)$ hold, we have $\Pi^{-1}_{I_{< d_j}}(j)>\Pi^{-1}_{I_{< d_j}}(i)$ if $d_i>d_j$, and $\Pi^{-1}_{I_{< d_i}}(j)>\Pi^{-1}_{I_{< d_i}}(i)$ otherwise.
\end{definition}

\begin{definition}
A permutation function $\Pi$ is compatible if for any agents $i, j, k\in N$, such that $\delta^{-1}(i)>\delta^{-1}(k)$, $\delta^{-1}(j)>\delta^{-1}(k)$, $\Pi^{-1}_{I_{< d_k}}(j)>\Pi^{-1}_{I_{< d_k}}(i)$ and $\Pi^{-1}_{I_{< d_k}}(k)>\Pi^{-1}_{I_{< d_k}}(i)$ hold, we have $\Pi^{-1}_{I_{< d_j}}(j)>\Pi^{-1}_{I_{< d_j}}(i)$ if $\delta^{-1}(i)>\delta^{-1}(j)$, and $\Pi^{-1}_{I_{< d_i}}(j)>\Pi^{-1}_{I_{< d_i}}(i)$ otherwise.
\end{definition}

In other words, a permutation is compatible if for any two agents $i$ and $j$, that are steel in the market after agent $k$ leaves, and such that $i$ appears earlier than $j$, who appears earlier than $k$ in the permutation when $k$ leaves the market, then agent $i$ will steel appear before $j$ according to the permutation function during the earliest departure between $i$ and $j$.

\end{comment}

\begin{proposition}\label{ICprop}
The exchange algorithm described in Algorithm \ref{algoOnlinePickingSequence} is $WIC$ if and only if the permutation function $\Pi$ is $PFE$.
\end{proposition}

\begin{proof}
Assume first that $\Pi$ is not $PFE$. This means that there is an instance $I$ and three agents $i$, $j$, and $k$ such that $k$ leaves earlier than $i$ and $j$, $\Pi(I_{< d_k})$ orders $i$ before $j$ and $k$, whereas either $\Pi(I_{< d_j})$ or $\Pi(I_{< d_i})$ orders $j$ before $i$, depending on who between $j$ and $i$ leaves first. Without loss of generality, we can assume that the preferences of the agents for this instance are as follows.\footnote{Note that we assumed earlier in the paper that the outcome of a permutation function is independent of the preferences of the agents.} The preferences of agent $i$ are $e_k\succ_i e_i \succ_i e_j \succ_i \ldots$, where all the other items are ranked arbitrarily after these three items. The preferences of agent $j$ are $e_k\succ_j e_j \succ_j e_i \succ_j \ldots$, and the preferences of agent $k$ are $e_i\succ_k e_k \succ_k e_j \succ_k \ldots$. Finally, the preferences $\succ_{\ell}$ of any other agent $\ell$ are such that $e_{\ell}$ appears in first position, whereas the other items appear in an arbitrary order.

If we apply Algorithm \ref{algoOnlinePickingSequence} to this instance $I$, then each other agent $\ell$ other than $i$, $j$, and $k$ will always choose her own item. Furthermore, when agent $k$ leaves the market, agent $i$, who chooses before agents $k$ and $j$ (if $j$ has already arrived at $d_k$), reserves $e_k$. Independently of the presence of agent $j$ in $I_{< d_k}$, or her choice if she is ranked before agent $k$ in $\Pi(I_{< d_k})$, agent $k$ selects $e_i$, and $e_i$ is assigned to agent $k$ when she leaves the market. Then, in both cases where either agent $i$ or agent $j$ leaves first, agent $j$ selects before agent $i$ and chooses $e_k$. When agent $i$ leaves the market, she receives $e_j$, her third choice.

Now, assume that agent $i$ reveals other preferences, say $\succ'_i$, instead of her true preferences. These new preferences are such that $e_i \succ'_i e_k \succ'_i e_j \succ'_i \ldots$. Let $I'$ denote the copy of $I$ where agent $i$ reveals $\succ'_i$ instead of $\succ_i$, and the other agents keep the same preferences. If we apply Algorithm \ref{algoOnlinePickingSequence} to $I'$, then agent $i$ reserves $e_i$ instead of $e_k$ when agent $k$ leaves the market, and agent $k$ picks $e_k$ or $e_j$, depending on the position of $j$ in $\Pi(I_{< d_k})$. Then, at time $\min\{d_i,d_j\}$, when either $i$ or $j$ leaves the market, agent $j$ selects either $e_k$ or $e_j$, depending on the item assigned earlier to agent $k$, and agent $i$ receives ultimately $e_i$, which is strictly better than $e_j$ according to her true preferences $\succ_i$. Therefore, agent $i$ has an incentive to misreport her preferences.

Now, we are going to show that if $\Pi$ is $PFE$, then Algorithm \ref{algoOnlinePickingSequence} is $WIC$. To do so, we first show that for any agent $i$, the set of items selected by the agent ranked before $i$ in $\Pi$ will not change if agent $i$ changes her revealed preferences. Let's consider the first iteration of the ``for'' loop of Algorithm \ref{algoOnlinePickingSequence} where agent $i$ has to select an item. Note that this iteration will occur at the same time whatever preferences are revealed by agent $i$. If this is the iteration corresponding to $d_i$, where agent $i$ leaves the market, then she has no incentive to misreport her preferences. Therefore, we can focus on the cases where the iteration does not correspond to $d_i$, and agent $i$ remains in the market in the next iterations. Let $K$ denote the set of items remaining in the market and not selected by one of the agents ranked before agent $i$ in $\Pi$. We are going to show that in the next iteration, no agent ranked before $i$ in $\Pi$ will select an item of $K$.

Let us show this by induction on the relative rank of the agent among the ones reserving before $i$. We start with the first agent, who selects her most favorite item among the remaining ones. During the previous iteration, she reserved an item which is not part of $K$. So this item must still be available. During this new iteration, she may prefer another item which was not present during the previous iteration, and leave unreserved the item that she chose before, but she will not select an item from $K$. 

Now, let us consider that the first $k-1$ ranked before agent $i$ did not select an item from $K$, and let us show that the agent at rank $k$ also does not select an item from $K$. Note that each of the $k-1$ first agents reserved during this iteration either the item that they reserved during the previous iteration, or an item that arrived later than the previous iteration, or an item reserved during the previous iteration by an agent ranked at a better position but who chose another item during this iteration and left the previously reserved item free. This means that the item reserved during the previous iteration by the agent at position $k$ is still available, and she can reserve it again, or reserve any other item that was not available during her turn in the previous iteration that she prefers. Thus, she will not select an item from $K$. 

This completes the induction and shows that no agent ranked before $i$ in $\Pi$ will be interested in selecting an item from $K$ in the next iteration of the ``for'' loop of Algorithm \ref{algoOnlinePickingSequence}. By applying the same reasoning, and by considering for each new iteration a larger set $K$ of items not chosen by the agents ranked before $i$, we can show that from iteration to iteration, no agent ranked before $i$ will select an item accessible to $i$ when it is her turn to reserve. This implies that the choice made by $i$ on the item that she reserves during each iteration has no consequence on the set of items selected by the agents who are ranked before her in $\Pi$. Therefore, she has no incentive to misreport her true preferences in order to reserve another item instead of her most favorite one, because if she does so, another agent ranked after her may pick it and leave the market with it.
\end{proof}

Note that, in essence, the definition of $PFE$ is quite similar to Definition~\ref{defPS} of prefix-stability. However, to establish the equivalence in Proposition~\ref{ICprop}, we had to extend the definition to permutations where the positions of agents who have already left the market are not fixed (since they are ignored by Algorithm~\ref{algoOnlinePickingSequence}). It is therefore easy to verify that the notion of $PS$ is slightly stronger than that of $PFE$.

\begin{observation}\label{obsSOCimpliesDOC}
    Any permutation function that is $PS$ is also $PFE$.
\end{observation}

This implies that all the permutation functions encountered earlier in the paper, including the ascending departure permutation $\delta$ and the ascending arrival permutation $\alpha$, render Algorithm \ref{algoOnlinePickingSequence} $WIC$. This is not surprising for the ascending departure permutation $\delta$ since Algorithm \ref{algoOnlinePickingSequence} coincides with Algorithm \ref{algoPickingSequence} when the ascending departure permutation $\delta$ is used as input. The following property shows that Algorithm \ref{algoOnlinePickingSequence} using the ascending arrival permutation $\alpha$ as input is also $d$-$IC$.

\begin{comment}
    
Let $\alpha$ denote the permutation where agents are ordered according to their arrival time\footnote{Ties are broken arbitrarily.} i.e., such that $a_{\alpha(1)} \leq a_{\alpha(2)} \leq \ldots \leq a_{\alpha(n)}$.

\end{comment}

\begin{proposition}\label{propdICalphadyn}
Algorithm \ref{algoOnlinePickingSequence} is $d$-$IC$ if and only if it uses the ascending arrival permutation $\alpha$ as input.
\end{proposition}

\begin{proof}
Let us first show that Algorithm \ref{algoOnlinePickingSequence} using the ascending arrival permutation $\alpha$ as input is $d$-$IC$. Note first that the ascending arrival permutation $\alpha$ is $PS$, and therefore by Proposition \ref{ICprop} and Observation \ref{obsSOCimpliesDOC}, we know that no agent has an incentive to misreport her preferences. So we only need to show that no agent has an incentive to misreport her departure time. Note also that as soon as the arrival does not change, changing the departure time of an agent has no impact on the ranking dictated by the ascending arrival permutation $\alpha$. We are going to show that if the departure time of an agent, say $i$, is delayed after a new event, which may be the departure of an agent or the arrival of a new agent, then her most favorite available item is either the same or even better.

More formally, let $I$ be an instance and $I'$ be the same instance except that the new departure time $d'_i$ of agent $i$ is postponed to be later than the arrival time $a_j$ or the departure time $d_j$ corresponding to the event occurring right after $d_i$ i.e., either \textit{(i)} $d_i<d_j<d'_i$ or \textit{(ii)} $d_i<a_j<d'_i$. In both cases \textit{(i)} and \textit{(ii)}, all iterations of the ``for'' loop of Algorithm \ref{algoOnlinePickingSequence} applied to instance $I$ and $I'$ before reaching $d_i$ are the same, since only the departure time of agent $i$ has changed.

In case \textit{(i)}, agent $j$ leaves earlier than $i$ in $I'$ whereas she left later in $I$. If $a_j > a_i$, then agent $i$ is ranked before agent $j$ according to the ascending arrival permutation $\alpha$. All the agents ranked before agent $i$ reserve the same items in $I$ and $I'$, so agent $i$ reserves in $I'$ the same item that she picks at iteration $d_i$ in $I$. Finally, all the agents ranked between agent $i$ and agent $j$ reserve the same items in $I'$ as the ones they reserve at $d_j$ in $I$. Then, agent $j$ picks the same item in $I'$ as the one she picks at $d_j$ in $I$. During the next iteration at time $d'_i$, all the agents ranked before agent $i$ reserve the same items in $I'$ as they reserve at $d_i$ in $I$, and agent $i$ selects the exact same item in $I'$ as the one she receives at $d_i$ in $I$. 

Consider now the other case where $a_j < a_i$. In this case, agent $i$ is not allowed to reserve an item during the iteration at $d_j$ in $I'$. Once again, during the iteration corresponding to $d_j$ in $I'$, all agents ranked before agent $j$ reserve the same items as in $I$, and therefore agent $j$ picks the item that she reserved in $I$. During the next iteration at $d'_i$, all agents ranked before agent $i$ reserve the same items in $I'$ as they do at $d_i$ in $I$. Although agent $j$ is no longer present to reserve an item at $d'_i$ in $I'$, she already picked the one she reserved at $d_j$ in $I$. Therefore, the set of items remaining for agent $i$ is the same at $d_i$ in $I$ and at $d'_i$ in $I'$, and she selects the same item.

In case \textit{(ii)}, a new item arrives to the market in $I'$, namely $e_j$, before agent $i$ leaves the market. This new item may affect the choices made by the agents ranked before agent $i$ in the ascending arrival permutation $\alpha$. Specifically, each of these agents may choose to reserve $e_j$ instead of the item they reserved in $I$ during the same step. Once item $e_j$ has been reserved, say by agent $k$, it is no longer available. However, the item that agent $k$ reserved in $I$ during the same iteration is now available and becomes an additional option in $I'$. As a result, each agent ranked before agent $i$ either reserves the same item as in $I$ or reserves $e_j$ (or another item left behind by an agent ranked before them). When it is agent $i$'s turn to select an item, the item assigned to her in $I$ is still available, along with an additional item, which could either be $e_j$ or another item abandoned by a preceding agent. Thus, agent $i$ will either pick the same item as in $I$ or a preferred item that was newly made available.

In summary, if agent $i$ delays her departure time, then she receives either the same item or a new one that she prefers. In other words, she has no incentive to report an earlier departure than her true departure, and the algorithm is $d$-IC.

Regarding the proof showing that the ascending arrival permutation $\alpha$ is the only permutation function rendering Algorithm \ref{algoOnlinePickingSequence} $d$-IC, it is essentially the same as the proof of Proposition \ref{alphaOnlyDIC} for Algorithm \ref{algoPickingSequence}.
\end{proof}

\begin{example}\label{exDynSerDict}
    Consider the instance described in the left part of Figure \ref{figOHMexample3}, and the three allocations graphically described in the right part of the same figure. It is easy to check that the static serial dictatorship procedure with the ascending departure permutation $\delta$ produces allocation $M$, and the same procedure with the ascending arrival permutation $\alpha$ produces allocation $M'$.
    
    Let us now run the dynamic serial dictatorship procedure with the ascending arrival permutation $\alpha$. At $d_2$, agent 1 is the first to choose, and she reserves $e_3$. Then, agent 2 is the second to choose, and she picks $e_1$. Item $e_1$ is therefore allocated to agent 2, and all other items remain unallocated. At $d_1$, agent 1 is the first to choose once again and she picks $e_4$ this time, which was not available at $d_2$. Therefore, item $e_4$ is allocated to agent 1. At $d_3$, agent 3 is the first to choose, and she picks item $e_2$. Finally, at $d_4$, agent 4 picks the only remaining item, $e_3$. The resulting allocation is therefore $M''$, as displayed in Figure \ref{figOHMexample3}. Note that $M''$ Pareto-dominates $M'$, which is the result of the static serial dictatorship procedure using the same ascending arrival permutation $\alpha$. However, this will not necessarily be the case each time, even though the result of the dynamic version of serial dictatorship will never be Pareto-dominated by the result of the static one.

    \begin{figure}[t]
\includegraphics[width=\textwidth]{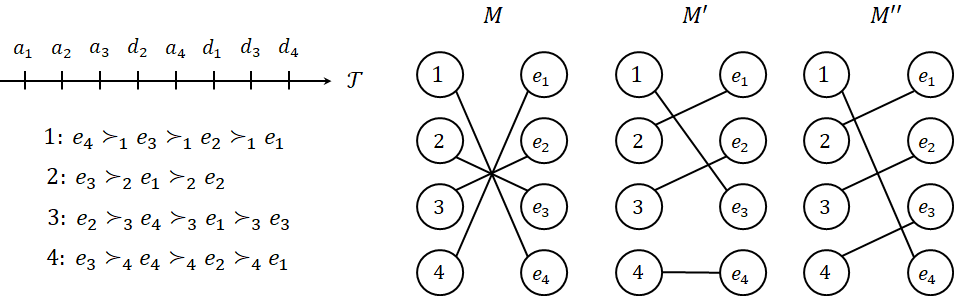}
\caption{Online housing market example with four agents.}\label{figOHMexample3}
\end{figure}
\end{example}

Similarly to Algorithm \ref{algoPickingSequence}, Algorithm \ref{algoOnlinePickingSequence} cannot be made $a$-IC.

\begin{proposition}\label{propDynSerDictIncompAIC}
There is no online permutation that makes Algorithm \ref{algoOnlinePickingSequence} $a$-IC.
\end{proposition}

\begin{proof}
By contradiction, assume that there exists a permutation $\Pi$ that makes Algorithm \ref{algoOnlinePickingSequence} $a$-IC. Consider the following instance $I$ with 3 agents, such that $a_1 < a_2 < d_1 < a_3 < d_2 < d_3$ hold. The preferences of the agents are $e_3 \succ_1 e_2 \succ_1 e_1$, $e_3 \succ_2 e_2 \succ_2 e_1$, and $e_3 \succ_3 e_2 \succ_3 e_1$. Note that since the algorithm is $a$-IC, agent 3 should be assigned her own item, which is also her most favorite one, since otherwise, by changing her arrival time to $a'_3$ such that $d_2 < a'_3$, she would be alone in the market at her arrival and would keep $e_3$.

 Let us consider instance $I_{< d_1}$ where agent 3 is removed. Once again, because the algorithm is $a$-$IC$, agent 2 should receive item $e_2$, since otherwise she can delay her arrival time after the departure of agent 1 and keep her item. This implies that item $e_1$ is assigned to agent 1 in $I_{< d_1}$. Note that this implies that $\Pi_{1}(I_{< d_1})=2$ and $\Pi_{2}(I_{< d_1})=1$ hold since $e_2$ is preferred to $e_1$ by agent 1 and would be chosen if it is not reserved by agent 2. Furthermore, since $e_3$ should be assigned to agent 3 in $I$ and $e_3$ is the most favorite item of agent 2, agent 3 should be able to reserve it, and therefore be ranked before agent 2 in $\Pi(I_{< d_2})=\Pi(I)$. But the fact that $\Pi$ is $a$-$IC$ implies that it is also $WIC$, and by Proposition \ref{ICprop} we know that $\Pi$ should be $PFE$. This leads to a contradiction since in the instance described earlier, we had agent 2 ranked before agent 1 in $\Pi_{I_{< d_1}}$, whereas agent 2 is ranked after an agent arriving later than $d_1$, namely agent 3, in $\Pi(I_{< d_2})$.

\end{proof}

\section{Safe Serial Dictatorship Procedure}\label{secSafeSerDict}

So far, we have considered and designed online algorithms consistent with most of the properties presented in Section \ref{secPrem}, with the exception of $IR$. In this subsection, we explore the possibility to achieve $IR$ using the serial dictatorship procedure. Note that Proposition \ref{propOnlyPO} states that the only algorithm achieving Pareto efficiency corresponds to Algorithm \ref{algoPickingSequence} (or equivalently Algorithm \ref{algoOnlinePickingSequence}) using the ascending departure permutation $\delta$ as input. As suggested in Section \ref{secAscendingArrival} (see footnote \ref{footnotePOnotIR}), this algorithm is neither $a$-$IC$, $d$-$IC$ nor $IR$, leading to the following corollary.

\begin{corollary}\label{corPOvsOthers}
There is no online exchange algorithm that both returns a $\mathcal{M}(I)$-Pareto optimal allocation for any instance $I \in \mathcal{I}$, and is either $a$-$IC$, $d$-$IC$ or $IR$.
\end{corollary}

This incompatibility with $IR$ is mainly due to the fact that in order to ensure Pareto efficiency in all cases, there is a need to risk at some point a violation of $IR$. In other words, among the full set of partial allocations, there are ones which may lead to a violation of $IR$ if no additional agents with appealing items appear later on. This risk can be highly beneficial if these hypothetical agents arrive with the right timing. However, this risk cannot be made in an online context to achieve $IR$ in all cases, since in the worst case, these agents will not come to counterbalance unfavorable exchanges. Therefore, in order to relax the notion of Pareto efficiency, we define a subset of allocations, called \textit{safe allocations}, that will never lead to $IR$ violation in the following way.

\begin{definition}
For any instance $I\in \mathcal{I}$, the set of safe allocations $\mathcal{S}(I)$ contains any allocation $M$ such that $M(i) \succeq_i e_i$ holds for any agent $i$, and there exists an allocation $M'$ for instance $I_{< d_i}$ such that $M'(j)=M(j)$ holds for any agent $j$ with $d_j\leq d_i$, and $M'(k)\succeq_k e_k$ holds for any other agent $k$ of $I_{< d_i}$.
\end{definition}

The definition of safe allocation means that after the departure of any agent, it is always possible to complete the current partial allocation to the agents that remain in the market, and to achieve individual rationality without considering agents that arrive later and their initial endowment. Algorithm \ref{algoIROnlinePickingSequence} is an extension of Algorithm \ref{algoOnlinePickingSequence} that focuses on safe allocations instead of any possible allocations. By doing so, the partial allocation remains safe, and the algorithm is individually rational.\footnote{A similar algorithm was considered in~\cite{BiroKP22}, where agents select their items sequentially, but in an offline setting.}

The procedure $bestSafe$ is similar to the procedure $best$ of Algorithms \ref{algoPickingSequence} and \ref{algoOnlinePickingSequence}, except that the best item is taken among the items whose allocation to the leaving agent does not make the partial allocation violate the property of safeness, i.e., it is always possible to complete this partial allocation to obtain a safe one. Note that this procedure can be computed in polynomial time by considering if each remaining item is safe to be allocated to the leaving agent. The problem of checking safeness can be modeled as the problem of searching a maximal matching in a classical assignment problem of the remaining agents to the remaining items, where an item cannot be assigned to an agent that prefers her initial endowment to it. Such a problem can be solved in polynomial time by using, e.g., the Hungarian algorithm \cite{lovasz1986matching,Burkard08}.

\begin{algorithm}[tb]
\caption{Safe serial dictatorship procedure}
\textbf{Input}:  Permutation function $\Pi$.

\begin{algorithmic}[1]\label{algoIROnlinePickingSequence}
\STATE $A\leftarrow \emptyset$. \COMMENT{Items already assigned.}
\STATE $B\leftarrow \emptyset$ \COMMENT{Agents already matched.}
 \FOR{$i=1$ \textbf{to} $n$}
   \STATE \COMMENT{Iteration occurring at $d_{\delta(i)}$.}
   \STATE Let $\pi$ denote $\Pi(I_{< d_{\delta(i)}})$.
   \STATE $j \leftarrow 1$
   \STATE $T\leftarrow \emptyset$ \COMMENT{Set of reserved items.}
   \WHILE{$\pi(j) \neq \delta(i)$}
     \IF{$\pi(j)\not\in B$}
       \STATE $T \leftarrow T\cup \{ bestSafe_{\pi(j)}(E_{< d_{\delta(i)}}\setminus (A\cup T))\}$.
     \ENDIF
     \STATE $j\leftarrow j+1$
   \ENDWHILE
     \STATE $M(\delta(i))\leftarrow bestSafe_{\delta(i)}(E_{< d_{\delta(i)}}\setminus (A\cup T))$.
     \STATE $A \leftarrow A\cup \{ M(\delta(i))\}$.
     \STATE $B \leftarrow B\cup \{ \delta(i)\}$.
 \ENDFOR
\STATE \textbf{return} $M$.
\end{algorithmic}
\end{algorithm}

By choosing the ascending departure permutation $\delta$ as input, Algorithm \ref{algoIROnlinePickingSequence} achieves a relaxed version of Pareto efficiency, as shown by the following proposition.

\begin{proposition}\label{propSafePO}
An online exchange algorithm is $IR$ and returns an $\mathcal{S}(I)$-$PO$ allocation for any instance $I\in \mathcal{I}$ if and only if it is Algorithm \ref{algoIROnlinePickingSequence} using the ascending departure permutation $\delta$ as input.
\end{proposition}

\begin{proof}
Let us first show that Algorithm \ref{algoIROnlinePickingSequence} using the ascending departure permutation $\delta$ as input is $IR$ and returns an $\mathcal{S}(I)$-$PO$ allocation for any instance $I$. The property of $IR$ is easy to show since, after each iteration of the ``for'' loop, the partial allocation $M$ can always be completed to be $IR$. This holds even if new agents arrive between two iterations, as their initial endowments can be assigned to them. Now, let us focus on the proof of the other property. The main argument is that the matching $M$ returned by the algorithm should be lexicographically optimal among the matchings in $\mathcal{S}(I)$, according to the order provided by the ascending departure permutation $\delta$ applied to $I$. Indeed, when the ascending departure permutation $\delta$ is used, there is no iteration during the while loop, and $T$ is empty when an agent chooses her item. Therefore, she chooses the best item among the items of $E_{< d_{\delta(i)}} \setminus A$ that maintains safeness. Let $\mathcal{P}$ denote the subset of safe allocations that may Pareto-dominate $M$. According to Algorithm \ref{algoIROnlinePickingSequence}, $M(\pi(1))$ should be the most favorite item of agent $\pi(1)$ in $E_{< d_{\pi(1)}}$ that maintains safeness. Therefore, all matchings in $\mathcal{P} \subseteq \mathcal{S}(I)$ should allocate $M(\pi(1))$ to agent $\pi(1)$ to Pareto-dominate $M$. Given that $M(\pi(1))$ should be allocated to agent $\pi(1)$, Algorithm \ref{algoIROnlinePickingSequence} allocates to agent $\pi(2)$ her most favorite item in $E_{< d_{\pi(2)}} \setminus \{M(\pi(1))\}$ that maintains safeness. Therefore, all matchings in $\mathcal{P}$ should allocate $M(\pi(2))$ to agent $\pi(2)$ to Pareto-dominate $M$. By applying the same reasoning along the elements of permutation $\pi$, we can reach the conclusion that $\mathcal{P}$ may only contain a single allocation corresponding to $M$. However, a matching cannot Pareto-dominate itself, and therefore no matching of $\mathcal{S}(I)$ Pareto dominates $M$.

Let us show now that any other online exchange algorithm cannot be both $IR$ and returns an $\mathcal{S}(I)$-$PO$ allocation for any instance $I$. We show this statement by contradiction by assuming that there exists another online exchange algorithm $\mathcal{A}$ that is both $IR$ and returns an $\mathcal{S}(I)$-$PO$ allocation for any instance $I$, and which behaves differently than Algorithm \ref{algoIROnlinePickingSequence} using the ascending departure permutation $\delta$ as input. This means that there is an instance $I$ and an agent $j$ such that $\mathcal{A}_j(I)$ does not correspond to the assignment made by Algorithm \ref{algoIROnlinePickingSequence} with the ascending departure permutation $\delta$ as input. In other words, agent $j$ does not receive her most favorite item among the remaining ones of $E_{< d_{\delta(j)}}$ that maintain safeness. Furthermore, we assume without loss of generality that agent $i$ is such an agent with the earliest departure time. Let $M$ denote the allocation obtained by Algorithm \ref{algoIROnlinePickingSequence} applied to $I$ and using the ascending departure permutation $\delta$ as input. Note that if $\mathcal{A}_j(I) \succ_j M(j)$ holds, then $\mathcal{A}_j(I)$ is not safe and algorithm $\mathcal{A}$ is not $IR$ since $\mathcal{A}(I_{< d_j})$ cannot assign to each agent an item as preferred as her initial endowment, concluding the proof. Therefore, we can focus on the case where $M(j) \succ_j \mathcal{A}_j(I)$. Let $i$ denote the agent who receives $M(j)$ instead of agent $j$ in $\mathcal{A}(I)$, i.e., such that $\mathcal{A}_i(I) = M(j)$. By assumption, $d_j < d_i$ should hold, since otherwise an agent arriving earlier than agent $j$ would not receive the same item as in $M$.
 
To show the contradiction, we construct a new instance $I'$ by adding dummy agents to $I$ arriving later than the departure of agent $j$. For each agent $k$ leaving after $d_j$, let $k'$ be a dummy agent such that $a_{k'} = d_{j} + k\epsilon$, with $\epsilon > 0$ small enough for each dummy agent to arrive earlier than any other departure, and $d_{k'}$ large enough to leave after any other agent of $I$. The preferences of the agents are as follows. For any agent $k$ that was already in $I$ and leaving after $d_j$, her most favorite item becomes $e_{k'}$, followed by the items of $I$ ranked in the same order as they were in her preferences in $I$, and finally all the other dummy items of $I'$ ranked in arbitrary order. For dummy agent $i'$ (corresponding to agent $i$), her preferences are $\mathcal{A}_j(I) \succ_{i'} M(j) \succ_{i'} \ldots$, where all the other items are ranked in arbitrary order. Finally, for any other dummy agent $k'$, her most favorite item is $\mathcal{A}_k(I)$, and all other items are ranked arbitrarily.

Note first that since $\mathcal{A}$ is online, both $\mathcal{A}_j(I')=\mathcal{A}_j(I'_{< d_j})$ and $\mathcal{A}_j(I_{< d_j})=\mathcal{A}_j(I)$ hold. Furthermore, since all dummy agents arrive later than $d_j$, $I_{< d_j}=I'_{< d_j}$ holds, which implies with the former equalities that $\mathcal{A}_j(I')=\mathcal{A}_j(I)$. Now concerning the assignment of the other agents, based on their preferences and since $\mathcal{A}$ returns an $\mathcal{S}(I')$-$PO$ allocation, we can say that for each dummy agent $k'$, we have $\mathcal{A}_{k'}(I')=\mathcal{A}_k(I)$, and for each regular agent $k$ leaving after $d_j$, we have $\mathcal{A}_k(I')=e_{k'}$. More precisely for $i'$ we have $\mathcal{A}_{i'}(I')=\mathcal{A}_i(I)=M(j)$. Note that except for agents $i'$ and $j$, any agent receives her most favorite item.

But now consider the almost same allocation, say $M'$, as the one constructed by $\mathcal{A}$, except that agent $j$ receives $M(j)$ and agent $i'$ receives $\mathcal{A}_j(I)$. Note that by construction, $M'$ should belong to $\mathcal{M}(I')$ since $\mathcal{A}$ is an online exchange algorithm, and therefore $\mathcal{A}(I_{<d_j})=\mathcal{A}(I'_{<d_j})$ belongs to $\mathcal{M}(I'_{<d_j})$ (for agent $j$ and the ones leaving after her), $\mathcal{A}(I')$ belongs to $\mathcal{M}(I')$ (for all other agents except agent $i'$), and $a_{i'}<d_i$ holds (for agent $i'$). Furthermore, $M'$ belongs also to $\mathcal{S}(I')$ since agent $j$ and each agent leaving before him receive the same item as in $\mathcal{A}(I)$, which is online, and during each later departure, the most favorite item of each agent remains in the market (corresponding to the ones they will ultimately receive). Agent $j$ and dummy agent $i'$ receive in $M'$ a strictly better item than the one offered by algorithm $\mathcal{A}$, since we have assumed that $M(j) \succ_j \mathcal{A}_j(I)$ holds. Therefore, $M'$ Pareto dominates $\mathcal{A}(I')$, leading to a contradiction since $\mathcal{A}$ should return an $\mathcal{S}(I')$-$PO$ allocation.
\end{proof}

\begin{example}\label{exSafeSerialDictator}
Consider the instance described in the left part of Figure \ref{figOHMexample4}, as well as the allocations graphically described in the right part of the same figure. It is easy to check that the static serial dictatorship procedure (Algorithm \ref{algoPickingSequence}) with the ascending departure permutation $\delta$ produces allocation $M$, whereas the dynamic serial dictatorship procedure (Algorithm \ref{algoOnlinePickingSequence}) with the ascending arrival permutation $\alpha$ returns $M'$.

Let us run the safe serial dictatorship with the ascending departure permutation $\delta$ on the same instance. Note first that this procedure, described in Algorithm \ref{algoIROnlinePickingSequence}, is essentially the same as Algorithm \ref{algoOnlinePickingSequence}, with the exception that the set of items an agent can pick is restricted to those which leave the possibility for the remaining agents to be allocated items at least as favorable as their initial endowment. In other words, the set of items is restricted to those such that, if they are matched to the agent, it is still possible to complete the allocation to obtain a complete individually rational allocation for all agents present in the market. Therefore, at $d_1$, agent 1 is the first to choose, but she cannot pick item $e_2$ because otherwise it would be impossible to allocate the remaining item $e_1$ to agent 2 without violating individual rationality. Hence, item $e_1$ is allocated to agent 1. At $d_3$, agent 2 is the first to choose, and she reserves item $e_3$. Note that she can do so because agent 3 prefers the remaining item $e_2$ to her initial endowment $e_3$. Then, agent 3 picks $e_2$, which is allocated to her. At $d_2$, agent 2 picks item $e_3$ once again, which is allocated to her. Note that she can do so because item $e_4$ remains available for agent 4. Finally, at $d_4$, agent 4 picks the last remaining item $e_4$, which is allocated to her. The resulting allocation $M''$ is graphically described in Figure \ref{figOHMexample4}.

\begin{figure}[t]
\includegraphics[width=\textwidth]{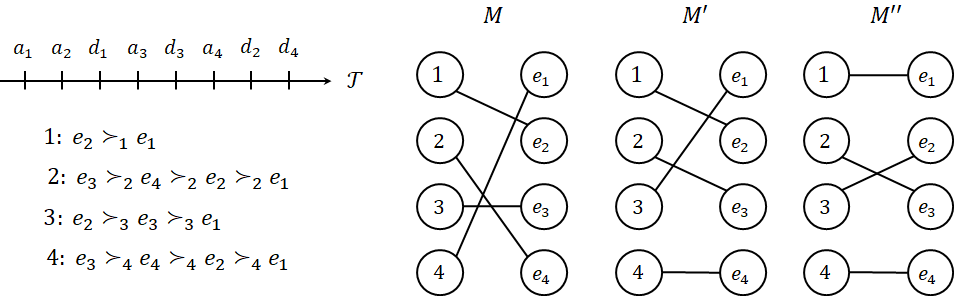}
\caption{Another online housing market example with four agents.}\label{figOHMexample4}
\end{figure}
\end{example}

It remains to explore the compatibility of Algorithm \ref{algoIROnlinePickingSequence} with the different notions of incentive compatibility. The following proposition shows that Algorithm \ref{algoIROnlinePickingSequence} is incompatible with any kind of incentive compatibility because agents may veto assignments by declaring some items as less desirable than their initial endowment.

\begin{proposition}\label{propSafeSerDictNoIC}
There is no permutation that makes Algorithm \ref{algoIROnlinePickingSequence} $WIC$.
\end{proposition}

\begin{proof}
Consider the following instance with three agents such that $a_1<a_2<a_3<d_1<d_2<d_3$. By contradiction, we assume that there is a permutation $\Pi$ that makes Algorithm \ref{algoIROnlinePickingSequence} $WIC$. To simplify notation, we denote by $(i)$ the agent of rank $i$ in $\Pi(I)$ (instead of $\Pi_i(I)$). The preferences of the agents are provided depending on their rank in $\Pi(I)$\footnote{Note that we assume that the rank of an agent in a permutation function does not depend on her preferences.}, and are $e_{(3)}\succ_{(1)} e_{(2)}\succ_{(1)} e_{(1)}$, $e_{(1)}\succ_{(2)} e_{(3)}\succ_{(2)} e_{(2)}$, and $e_{(1)}\succ_{(3)} e_{(2)}\succ_{(3)} e_{(3)}$.

Note that the preferences are designed such that all allocations are $IR$, which implies that Algorithm \ref{algoIROnlinePickingSequence} works exactly as Algorithm \ref{algoOnlinePickingSequence} in this specific case. Independently of their departure time, agents $(1)$, $(2)$ and $(3)$ will pick items $e_{(3)}$, $e_{(1)}$ and $e_{(2)}$, respectively. But if agent $(3)$ changes its preferences to $e_{(1)}\succ'_{(3)} e_{(3)}\succ'_{(3)} e_{(2)}$, then all allocations where she may receive item $e_{(2)}$ are considered as unsafe. In that case, agent $(1)$ will still be able to choose item $e_{(3)}$ during Algorithm \ref{algoIROnlinePickingSequence} using $\Pi$, but then agent $(2)$ won't be able to choose item $e_{(1)}$ since otherwise the only remaining item available to agent $(3)$ would be $e_2$. Therefore, agent $(2)$ has to pick item $e_{(2)}$, and agent $(3)$ receives item $e_{(1)}$. This means that agent $(3)$ has an incentive to misreport her true preferences, leading to a contradiction with the fact that Algorithm \ref{algoIROnlinePickingSequence} using permutation function $\Pi$ is $WIC$.
\end{proof}

\section{Top Trading Cycle Based Algorithms}\label{secTTClike}

In this section, we explore the possibility to extend the well-known top trading cycle algorithm to our online problem \cite{SHAPLEY197423}. The top trading cycle (TTC) procedure is a natural candidate to achieve the properties mentioned in Section \ref{secPrem}, since in the offline context it fulfills all of them, and even stronger requirements \cite{RePEc:eee:mateco:v:4:y:1977:i:2:p:131-137,ma1994strategy}. Before providing an extension of TTC to the online setting, we provide a description of the standard TTC procedure in Algorithm \ref{ttc}, which will be used as a subroutine of our online algorithms. This procedure relies on the TTC-graph, which is an oriented graph where agents are vertices. As described in Algorithm~\ref{algoTTCgraph}, which explains its construction, each vertex has out-degree one: the unique outgoing edge from the vertex corresponding to agent~$i$ points to the vertex corresponding to the agent~$j$ whose item~$e_j$ is, according to $\succ_i$, the most preferred among those held by the agents currently in the graph. Such a graph should contain at least one cycle, and items are assigned along these cycles. Each agent contained in the cycle is assigned the item belonging to the agent that she is pointing to and leaves the market with this item. The procedure continues repeatedly with the remaining agents and their initial endowment until each agent is assigned an item.

\begin{algorithm}[t]
\caption{TTC-graph construction}
\label{algoTTCgraph}
\textbf{Input}: Subset of agents $N'$ and their initial endowment.\\
\textbf{Output}: TTC-graph $G(N')$.
\begin{algorithmic}[1]
\STATE $V\leftarrow N'$. \COMMENT{Set of vertices.}
\STATE $H\leftarrow \bigcup_{i\in N'}\{e_i\}$.\COMMENT{Houses initially owned by the agents of $N'$.}
\STATE $E\leftarrow \emptyset$. \COMMENT{Set of edges.}
 \FOR{\textbf{each} $i$ \textbf{in} $N'$}
   \STATE $E\leftarrow E\cup (i, best_{i}(H))$.
 \ENDFOR
\STATE \textbf{return} $G=(V, E)$.
\end{algorithmic}
\end{algorithm}

\begin{algorithm}[tb]
\caption{Top trading cycle algorithm}
\textbf{Input}:  Subset of agents $N'$ and their initial endowments.

\begin{algorithmic}[1]\label{ttc}
\WHILE{$N'\neq \emptyset$}
  \STATE Construct the TTC-graph $G(N')$ using Algorithm \ref{algoTTCgraph}.
  \FOR{\textbf{each} cycle $C=(c_1, c_2, \ldots , c_k)$ \textbf{in} $G(N')$}
    \FOR{$i=1$ to $k-1$}
      \STATE $M'(c_i)\leftarrow e_{c_{i+1}}$.
    \ENDFOR
    \STATE $M'(c_k)\leftarrow e_{c_{1}}$.
    \STATE $N'\leftarrow N'\setminus C$.
  \ENDFOR
\ENDWHILE
\STATE \textbf{return} $M'$.
\end{algorithmic}
\end{algorithm}

It is well known that the TTC procedure is $IR$ in the offline setting, and any online extension of this algorithm should keep this property. However, concerning the other properties, in the online setting not all agents are present at the beginning of the algorithm, and some of them will not arrive before the departure of other agents. Therefore, Algorithm \ref{ttc} should be customized to fit the online setting. We explore a static (or greedy) procedure described in Algorithm \ref{algoGreedyTTC}, which applies TTC to subsets of agents already present in the market, according to an input partition function. As suggested by its name, partition function $\mathcal{P}$ applied to $I$ will partition the agents belonging to $I$ into multiple subsets. More formally, a partition function $\mathcal{P}: \mathcal{I} \to 2^{2^{N}}$ maps any instance $I \in \mathcal{I}$ to a partition $\mathcal{P}(I)$ of the agents in $I$. This means that $\mathcal{P}(I)$ is a collection of subsets of agents such that the union of these subsets is $N$, and no two subsets share an agent. Similarly to the permutation function used in the serial dictatorship procedure, we assume that its output is independent of the agents' preferences and initial endowments, and depends only on their arrival and departure times. To construct the allocation, Algorithm \ref{ttc} is applied independently to each subset of agents returned by this partition function. To simplify notations, we assume that $\mathcal{P}_i(I)$ denotes the subset of $\mathcal{P}(I)$ that contains agent $i$.

The static TTC procedure described in Algorithm \ref{algoGreedyTTC} requires that an agent is not selected multiple times for application of Algorithm \ref{ttc}. This translates into the following property.

\begin{definition}
A partition function $\mathcal{P}$ is progress-preserving ($PP$) if for any instance $I$ and agent $i$, $\mathcal{P}_i(I) = \mathcal{P}_i(I_{< d_i})$ holds.
\end{definition}

Note that this property also implies that for any instance $I$ and for any agent $i$, $\mathcal{P}_i(I)$ contains only agents whose arrival times are no later than $d_i$. Furthermore, $\mathcal{P}_i(I)$ cannot contain an agent $j$ whose departure time $d_j$ is earlier than $a_i$, since otherwise we would have $\mathcal{P}_i(I_{< d_i}) = \mathcal{P}_i(I) = \mathcal{P}_j(I) = \mathcal{P}_j(I_{< d_j})$, which is impossible, as agent $i$ is not part of $I_{< d_j}$. Therefore, $\mathcal{P}_i(I)$ contains only agents that are in the market at time $d_k=\min_{j\in \mathcal{P}_i(I)}d_j$, i.e., for any $j \in \mathcal{P}_i(I)$, $a_j<d_k\leq d_j$ must hold. Following the same line of reasoning, it is easy to check that for any agent $j$ belonging to $\mathcal{P}_i(I_{<d_i})$ and leaving later than $d_i$, we have $\mathcal{P}_j(I_{<d_i})=\mathcal{P}_i(I_{<d_i})=\mathcal{P}_i(I_{<d_j})$, which implies that $\mathcal{P}_j(I_{<d_j})=\mathcal{P}_i(I_{<d_j})$ since $j\in\mathcal{P}_j(I_{<d_i})=\mathcal{P}_i(I_{<d_j})$. Therefore, no agent participates multiple times to the TTC procedure when the partition function used by Algorithm \ref{algoGreedyTTC} is $PP$.
 The following proposition shows that Algorithm \ref{algoGreedyTTC} is online and agents have no incentive to misreport their preferences.

\begin{algorithm}[tb]
\caption{Static online top trading cycle procedure}
\textbf{Input}: Partition function $\mathcal{P}$.

\begin{algorithmic}[1]\label{algoGreedyTTC}
\STATE $B\leftarrow \emptyset$ \COMMENT{Agents already matched.}
 \FOR{$i=1$ \textbf{to} $n$}
   \STATE \COMMENT{Iteration occurring at $d_{\delta(i)}$.}
   \IF{$\delta(i)\not\in B$}
   \STATE Applies Algorithm \ref{ttc} to $\mathcal{P}_i(I_{< d_{\delta(i)}})$ to obtain partial matching $M'$.
   \FORALL{$j\in \mathcal{P}_i(I_{< d_{\delta(i)}})$}
     \STATE $M(j)\leftarrow M'(j)$.
   \ENDFOR
   \STATE $B\leftarrow B \cup \mathcal{P}_i(I_{< d_{\delta(i)}})$
  \ENDIF
 \ENDFOR
\STATE \textbf{return} $M$.
\end{algorithmic}
\end{algorithm}

\begin{proposition}\label{propGreedyIC}
Algorithm \ref{algoGreedyTTC} using a $PP$ partition function as input is both online and $WIC$.
\end{proposition}

\begin{proof}
Let usfirst show that Algorithm \ref{algoGreedyTTC}, denoted $\mathcal{A}$ in this proof, is online when the partition function at use is $PP$. Consider an instance $I$ and an agent $i$. Let $I'$ denote the restricted instance $I_{< d_i}$. We need to show that $\mathcal{A}_i(I)=\mathcal{A}_i(I')$. Since $a_i < d_j$ holds for any agent $j$ of $I'$, it follows that $I_{< d_j}=I'_{< d_j}$. This implies that $\mathcal{P}_j(I_{< d_j})=\mathcal{P}_j(I'_{< d_j})$ holds for any agent $j$ in $I'$. Therefore, the partitions of agents in $I'$ are the same in $\mathcal{P}(I)$ and $\mathcal{P}(I')$, and Algorithm \ref{algoGreedyTTC}, applied to both $I$ and $I'$, executes the same iterations up to the departure time $d_i$ of agent $i$. This implies that $\mathcal{A}_i(I)=\mathcal{A}_i(I')$.

Let us now show that it is $WIC$. It is well known that Algorithm \ref{ttc} is incentive compatible. During Algorithm \ref{algoGreedyTTC}, multiple instances of Algorithm \ref{ttc} are applied, but each agent participates in exactly a single one of them. Furthermore, we assumed that the partition function used to select the agents participating in each instance of Algorithm \ref{ttc} does not depend on the agents' preferences. Therefore, no agent has an incentive to misreport her preferences, and Algorithm \ref{algoGreedyTTC} is $WIC$.
\end{proof}

\subsection{The Departing Agent Excluded Partition}

We explore now if a stronger notion of incentive compatibility can be achieved with particular partition functions. Let us introduce the departing agent excluded partition $\gamma$, whose construction is described in Algorithm \ref{algoParitionGamma}. Intuitively, every time that an unpartitioned agent leaves the market, two new subsets are added to the partition. The first is a singleton containing the agent leaving the market, and the second is the remaining unpartitioned agents that are present in the market. The following proposition shows that this partition function is progress-preserving.

\begin{algorithm}[tb]
\caption{Construction of the departing agent excluded partition $\gamma$}
\textbf{Input}:  Instance $I=\{ (i, e_i, a_i, d_i, \succ_i)\}_{i\in N}$.

\begin{algorithmic}[1]\label{algoParitionGamma}
\STATE $B\leftarrow \emptyset$. \COMMENT{Set of agents already belonging to the partition.}
\STATE $\gamma \leftarrow \emptyset$ \COMMENT{The partition to construct.}
 \FOR{$i=1$ \textbf{to} $n$}
   \STATE \COMMENT{Iteration occurring at $d_{\delta(i)}$.}
   \IF{$\delta(i)\not\in B$}
     \STATE $\gamma \leftarrow \gamma \cup \{\{ \delta(i)\}\} \cup \{ N_{< d_{\delta(i)}} \setminus (B\cup \{ \delta(i)\})\}$
     \STATE $B\leftarrow N_{< d_{\delta(i)}}$
   \ENDIF
 \ENDFOR
\STATE \textbf{return} $\gamma$.
\end{algorithmic}
\end{algorithm}

\begin{proposition}\label{propGammaOnline}
    The departing agent excluded partition $\gamma$ is $PP$
\end{proposition}

\begin{proof}
    We need to show that for any instance $I$ and for any agent $i$, $\gamma_i(I)=\gamma_i(I_{< d_i})$ holds. Let $i'$ denote the rank of agent $i$ in the ascending departure permutation $\delta$, i.e., such that $\delta(i')=i$ holds. Note that for any position $j\leq i'$, agents in $N_{< d_{\delta(j)}}$ are all contained in $I_{< d_i}$ since $d_i\geq d_{\delta(j)}$. Therefore, the first $i'$ iterations of the ``for'' loop of Algorithm \ref{algoParitionGamma} are the same for $I$ and $I_{< d_i}$. Consequently, the subset of agents containing agent $i$ in the departing agent excluded partition $\gamma$ is the same for both instances $I$ and $I_{< d_i}$.
\end{proof}

\begin{example}\label{exTTC}
    Consider the instance $I$ depicted in the left part of Figure \ref{figOHMexample2}, and the three allocations described in the right part of the same figure. It is easy to verify that both Algorithms \ref{algoPickingSequence} and \ref{algoOnlinePickingSequence} using either the ascending departure permutation $\delta$ or the ascending arrival permutation $\alpha$ return allocation $M''$, which is not $IR$ since agent 5 receives item $e_1$.

In a first step, let us run Algorithm \ref{algoParitionGamma} on this instance in order to compute the departing agent excluded partition $\gamma(I)$. At $d_1$, $N_{<d_1}$ contains agents 1, 2, and 3. Therefore, $\{1\}$ and $\{2, 3\}$ are added to the partition. At $d_2$ and $d_3$, nothing occurs since agents 2 and 3 already belong to the partition. At $d_4$, $N_{<d_4}$ contains the five agents, but agents 1, 2, and 3 already belong to the partition. Therefore, $\{4\}$ and $\{5\}$ are added to the partition, and Algorithm \ref{algoParitionGamma} halts, returning the departing agent excluded partition $\gamma(I) = \{\{1\}, \{2, 3\}, \{4\}, \{5\}\}$.

    Let us now run Algorithm \ref{algoGreedyTTC} with the departing agent excluded partition $\gamma(I)$ as input to compute the allocation for this instance. The TTC procedure is applied only once to a subset of more than one agent, which is $\{2, 3\}$. As depicted in Figure \ref{figTTCGraph}, the TTC graph for the subset of agents $\{2, 3\}$ contains a single cycle covering the two agents, and they swap their goods. The resulting allocation, denoted by $M$, is illustrated in Figure~\ref{figOHMexample2}. Note that the same allocation is returned by Algorithm \ref{algoIROnlinePickingSequence} using the ascending departure permutation $\delta$, and therefore $M$ is $\mathcal{S}(I)$-$PO$.

    \begin{figure}[t]
\includegraphics[width=\textwidth]{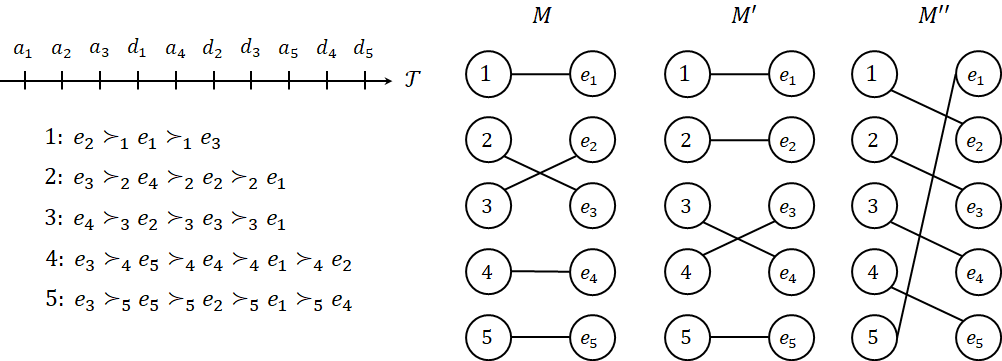}
\caption{Online housing market example with five agents.}\label{figOHMexample2}
\end{figure}

\begin{figure}[t]
\includegraphics{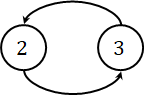}\centering
\caption{TTC-graph for the subset of agents $\{2, 3\}$.}\label{figTTCGraph}
\end{figure}
\end{example}

The idea under the departing agent excluded partition $\gamma$ is to ``punish'' agents, who force the assignment to take place because they leave earlier, by leaving them alone in the partition, with the consequence for Algorithm \ref{algoGreedyTTC} that they keep their initial endowment without exchange. The following proposition shows more formally that no agent has an incentive to misreport her departure time when Algorithm \ref{algoGreedyTTC} is used with the departing agent excluded partition $\gamma$.

\begin{proposition}\label{propGammaDIC}
    Algorithm \ref{algoGreedyTTC} using the departing agent excluded partition $\gamma$ is $d$-$IC$.
\end{proposition}

\begin{proof}
    Note first that we know by Propositions \ref{propGreedyIC} and \ref{propGammaOnline} that Algorithm \ref{algoGreedyTTC} using the departing agent excluded partition $\gamma$ is $WIC$. It remains to show that no agent has an incentive to misreport her departure time. Let $S$ denote the subset of agents for which the ``if'' clause of Algorithm \ref{algoParitionGamma} is true, i.e., agents that are not partitioned until their departure time.
    
    Consider first the case of an agent $i$ of $S$. Note first that agent $i$ belongs to a singleton in the departing agent excluded partition $\gamma$. Let us show that even if agent $i$ declares an earlier departure time $d'_{i}$ such that $a_{i} < d'_{i} < d_{i}$, then there is no other agent $j$ of $S$ such that $d'_{i} < d_{j} < d_{i}$ holds. Indeed, since $d_j < d_{i}$, we know that agent $j$ is considered before agent $i$ during the ``for'' loop of Algorithm \ref{algoGreedyTTC}. In the last instruction of this ``for'' loop when agent $j$ is considered, the set of agents $N_{< d_{j}}$ is added to the set of partitioned agents $B$. Since agent $i$ is not part of $B$ at iteration $d_i$ of the ``for'' loop, it means that she does not belong to $N_{< d_{j}}$, or in other words $a_{i} > d_{j}$ holds. This implies, together with $a_{i} < d'_{i}$, that $d_{j} < d'_{i}$ holds. 

We now know that any agent of $S$ that leaves earlier than agent $i$ will still leave earlier than $i$ even after $i$ declares a new departure time $d'_{i}$. This implies that agent $i$ will still pass the ``if'' test in the ``for'' loop of Algorithm \ref{algoParitionGamma} even if she changes her departure time to $d'_{i}$ because no agent of $S$ leaving earlier will include agent $i$ in a partition. Therefore, agent $i$ will still be part of a singleton in the departing agent excluded partition $\gamma$ after misreporting her departure time and will keep her own item during Algorithm \ref{algoGreedyTTC}. Consequently, agent $i$ has no incentive to misreport her departure time.

    Consider now the case of an agent $i$ that is not part of $S$. Note first that the fact that $i$ is not part of $S$ means that there exists an agent $j$ of $S$ such that $d_j<d_{i}$ and $i\in N_{< d_{j}}$ hold. We assume that $j$ is the agent of $S$ with the lowest rank in the ascending departure permutation $\delta$ for which these inequalities are true. In other words, the subset of agents containing agent $i$ is selected by Algorithm \ref{algoParitionGamma} during the iteration of the ``for'' loop corresponding to $d_j$. By applying the same reasoning as above, it is easy to check that there is no agent $k$ of $S$ such that $d_k<d_j$ and $d'_{i}<d_{k}$. Therefore, the only misreport of agent $i$ that could change the outcome of Algorithm \ref{algoParitionGamma} is a departure time such that $d_{j}>d'_{i}$ holds. In that case, agent $i$ will belong to the set of agents for which the ``if'' condition of Algorithm \ref{algoParitionGamma} is true, and therefore agent $i$ will be part of a singleton in the departing agent excluded partition $\gamma$, and will receive her own item in Algorithm \ref{algoGreedyTTC}. Since Algorithm \ref{algoGreedyTTC} is $IR$, agent $i$ has no incentive to misreport her departure time since she cannot receive a less preferred item than her own.  
\end{proof}

Unfortunately, Algorithm \ref{algoGreedyTTC} using the departing agent excluded partition $\gamma$ does not incentivize agents to reveal their true arrival time.

\begin{proposition}
   Algorithm \ref{algoGreedyTTC} using the departing agent excluded partition $\gamma$ as input is not $a$-$IC$.
\end{proposition}
    
\begin{proof}
    Consider the following instance $I$ with 6 agents such that $a_1<a_2<a_3<d_3<d_2<a_4<a_5<d_5<d_4<d_1$ hold. It is easy to check that $\gamma_1(I)=\{1, 2\}$. Assume that agent 1's most favorite item is $e_4$, whereas both agents 2 and 4's most favorite item is $e_1$. If agent 1 reveals her true arrival time $a_1$, then Algorithm \ref{algoGreedyTTC} using the departing agent excluded partition $\gamma$ will assign her either item $e_1$ or $e_2$, depending on agent 1's preferences for $e_2$. On the other hand, if she misreports her arrival time and reports $a'_1$ such that $d_2<a'_1<a_4$, then assuming that $I'$ denotes the copy of $I$ where agent 1 reveals $a'_1$ instead of $a_1$, it is easy to check that $\gamma_1(I')=\{1,4\}$. Therefore, according to the preferences of agent 1 and 4, item $e_4$ is assigned by Algorithm \ref{algoGreedyTTC} to agent 1. Since $e_4$ is the most favorite item of agent 1, she has an incentive to misreport her arrival time.
\end{proof}

\subsection{Partition Functions Based on Schedules}

Let us now examine if $a$-$IC$ can be achieved by Algorithm \ref{algoGreedyTTC} using another partition function. To do so, we need to introduce an additional element $\xi=\{\xi_0, \xi_1, \ldots, \xi_k\}$, that we call scheduling, which is a subset of non-overlapping time intervals. More formally, we assume that for any $j\in \{ 0, 1, \ldots, k\}$, $\xi_j\subseteq \mathcal{T}$ holds, and for any $j'\neq j$, $\xi_j\cap \xi_{j'}=\emptyset$ holds. We assume that each interval is contiguous, potentially with the exception of the first one $\xi_0$, which covers all time intervals not contained in the other intervals of the scheduling. This scheduling will be used to initialize the partition by grouping agents whose departure times belong to the same time interval of $\xi$.

The construction of the scheduled departure partition $\theta$, which is described in Algorithm \ref{algoParitionTheta}, depends on both the arrival and departure times of the agents, and the input scheduling $\xi$. Intuitively, all agents whose departure times belong to the same time interval of $\xi$ and whose arrival times are no later than the earliest departure time of one of these agents are grouped together, whereas the other agents are left alone in the partition. The first interval of $\xi$, denoted $\xi_0$, plays a special role since all the agents whose departure times belong to this interval are left alone in the partition. Note that this first interval may be left empty. To simplify notation, we denote by $\theta$ the scheduled departure partition resulting from Algorithm \ref{algoParitionTheta}, without specifying the scheduling used. This is not a big issue here since the properties proved later will not depend on the scheduling used as input. The following proposition shows that the partition function resulting from Algorithm \ref{algoParitionTheta} is progress-preserving.

\begin{algorithm}[tb]
\caption{The scheduled departure partition $\theta$ derived from scheduling $\xi$.}
\textbf{Input}:  Instance $I=\{ (i, e_i, a_i, d_i, \succ_i)\}_{i\in N}$, and scheduling $\xi=\{\xi_0, \xi_1, \ldots \xi_k\}$.

\begin{algorithmic}[1]\label{algoParitionTheta}
\STATE $A\leftarrow \{ 0\}$. \COMMENT{Set of time interval indices from $\xi$ that have already been considered.}
\STATE $B\leftarrow \emptyset$. \COMMENT{Set of agents already belonging to the partition.}
 \FOR{$i=1$ \textbf{to} $n$}
   \STATE \COMMENT{Iteration occurring at $d_{\delta(i)}$.}
   \STATE Let $\xi_j$ denote the interval of $\xi$ containing $d_{\delta(i)}$.
   \IF{$j\not\in A$}
     \STATE Let $S$ be the subset of agents of $N_{< d_{\delta(i)}}$ whose departure times belong to $\xi_j$.
     \STATE $\theta \leftarrow \theta \cup \{ S\}$.
     \STATE $A\leftarrow A \cup \{j\}$.
     \STATE $B \leftarrow B \cup S$.
   \ELSIF{$\delta(i)\not\in B$}
     \STATE $\theta \leftarrow \theta \cup \{\{ \delta(i)\}\}$.
     \STATE $B \leftarrow B \cup \{ \delta(i)\}$.
   \ENDIF
 \ENDFOR
\STATE \textbf{return} $\theta$.
\end{algorithmic}
\end{algorithm}

\begin{proposition}\label{propThetaOnline}
     The scheduled departure partition $\theta$ is $PP$.
\end{proposition}
    
\begin{proof}
    We need to show that for any instance $I$ and for any agent $i$, $\theta_i(I)=\theta_i(I_{< d_i})$ holds. Let $i'$ denote the rank of agent $i$ in the ascending departure permutation $\delta$, i.e., such that $\delta(i')=i$ holds. Note that for any position $j\leq i'$, agents whose arrival times are no later than $d_{\delta(j)}$ are all contained in $I_{< d_i}$. Therefore, the first $i'$ iterations of the ``for'' loop are the same in $I$ and $I_{< d_i}$. Thus, the subset of agents containing agent $i$ in the scheduled departure partition $\theta$ is the same for both instances $I$ and $I_{< d_i}$.
\end{proof}

\begin{example}\label{exSchedulePartition}
    Consider once again instance $I$ described in the left part of Figure \ref{figOHMexample2}, and scheduling $\xi$ such that $[a_3, d_2] \subseteq \xi_1$ and $[d_3, d_5] \subseteq \xi_2$ hold. The timeline is reproduced in Figure \ref{figScheduling} to provide a graphical representation of the scheduling $\xi$. Note that $\xi_1$ contains the departure time of agents 1 and 2, and $\xi_2$ contains the departure time of agents 3, 4, and 5. Let us run Algorithm \ref{algoParitionTheta} to compute the scheduled departure partition $\theta(I)$. At $d_1$, both agents 1 and 2 are in $N_{<d_1}$. Therefore, $\{1, 2\}$ is added to the partition. At $d_2$, nothing happens since agent 2 already belongs to the partition. At $d_3$, both agents 3 and 4 are in $N_{<d_3}$. Therefore, $\{3, 4\}$ is added to the partition. At $d_4$, nothing happens since agent 4 is already part of the partition. At $d_5$, the interval $\xi_2$ that contains $d_5$ has already been considered. Therefore, agent 5 remains alone, and $\{5\}$ is added to the partition. Finally, the algorithm halts and returns the scheduled departure partition $\theta(I) = \{\{1, 2\}, \{3, 4\}, \{5\}\}$.

\begin{figure}[t]
\includegraphics[width=\textwidth]{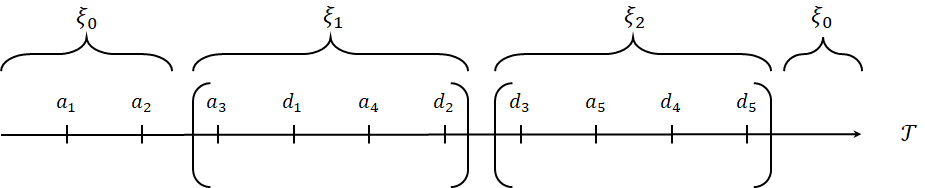}
\caption{Graphical representation of scheduling $\xi$.}\label{figScheduling}
\end{figure}
    
    Let us now run Algorithm \ref{algoGreedyTTC} using the scheduled departure partition $\theta(I)$ as input to compute the allocation. The TTC procedure is applied independently to $\{1, 2\}$ and $\{3, 4\}$. The TTC-graphs during the first iteration of the TTC procedure are represented in Figure \ref{figTTCGraph2}. In the left part of the figure, we can see the TTC-graph for the subset of agents $\{1, 2\}$, which contains a single cycle, the self-loop involving agent 2. Therefore, agent 2 keeps her good, and agent 1, who remains alone during the second iteration, also keeps her good. The TTC-graph shown in the right part of the figure corresponds to the subset of agents $\{3, 4\}$. This TTC-graph contains a single cycle that includes both agents, who swap their goods. The resulting matching $M'$ is depicted in the right part of Figure \ref{figOHMexample2}.

    \begin{figure}[t]
\includegraphics{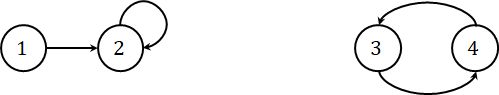}\centering
\caption{TTC-graphs for the subsets of agents $\{1, 2\}$ (left part) and $\{3, 4\}$ (right part).}\label{figTTCGraph2}
\end{figure}
\end{example}

Note that if two agents $i$ and $j$ have both their departure times belonging to the same time interval, say $\xi_k$, and one of them, say $i$, leaves before the arrival of the other, then agent $j$ is left alone in the scheduled departure partition $\theta$. That was the case for agent 5 in the preceding example. The idea behind this is that the set containing agent $i$ in the scheduled departure partition $\theta$ must be decided before her departure time, and agent $j$, who arrives too late, cannot be included in the same set. Additionally, agent $j$ cannot be included in another set of agents to avoid incentivizing her to strategically delay her arrival time. The following proposition shows more formally that no agent has an incentive to misreport her arrival time when Algorithm \ref{algoGreedyTTC} is used with a partition function returned by Algorithm \ref{algoParitionTheta}.

\begin{proposition}\label{propThetaAIC}
    Algorithm \ref{algoGreedyTTC} using the scheduled departure partition $\theta$ is $a$-$IC$.
\end{proposition}

\begin{proof}
    Note first that we know by Propositions \ref{propGreedyIC} and \ref{propThetaOnline} that Algorithm \ref{algoGreedyTTC} using the scheduled departure partition $\theta$ as input is $WIC$. It remains to show that no agent has an incentive to misreport her arrival time. Note first that since the departure time should be revealed truthfully by each agent, no agent can change the time interval of $\xi$ that is associated with her. Note also that if an agent's departure time belongs to $\xi_0$, then whatever her reported arrival time is, she will be alone in the scheduled departure partition $\theta$ and keep her own item during Algorithm \ref{algoGreedyTTC}. Therefore, she has no incentive to misreport her true arrival time. We can now focus on agents whose arrival time does not belong to $\xi_0$.
    
    For each agent $i$, there are two possible cases. Let $j$ denote the agent whose departure time is the earliest among the ones whose arrival times belong to the same interval of $\xi$ as agent $i$. The two possible cases are either \textit{(i)} $d_j<a_i$ or \textit{(ii)} $d_j>a_i$ holds. In case \textit{(i)}, according to Algorithms \ref{algoGreedyTTC} and \ref{algoParitionTheta}, agent $i$ will keep its own item since she does not belong to $I_{< d_j}$. It is easy to check that if she misreports her arrival time to later than $a_i$ then the situation will remain the same for her, and therefore she has no incentive to misreport her arrival time in that case. In case \textit{(ii)}, according to Algorithms \ref{algoGreedyTTC} and \ref{algoParitionTheta}, agent $i$ will be able to exchange her item with a subset of agents of $I_{< d_j}$, that is potentially not a singleton, during a TTC procedure applied to this subset of agents. If agent $i$ misreports her arrival time to no later than $d_j$ then the outcome is the same for her. However, if she misreports to a later time than $d_j$, then she will be left alone in the scheduled departure partition $\theta$ and keep her own item. Since the TTC procedure is $IR$, agent $i$ has no incentive to misreport her arrival time in this case. Therefore, in both cases \textit{(i)} and \textit{(ii)}, agent $i$ has no incentive to misreport her arrival time, and the procedure is $a$-$IC$.
 
\end{proof}

Unfortunately, Algorithm \ref{algoGreedyTTC} using the scheduled departure partition $\theta$ does not incentivize agents to reveal their true departure times.

\begin{proposition}
   Algorithm \ref{algoGreedyTTC} using the scheduled departure partition $\theta$, based on a scheduling containing not only $\xi_0$, is not $d$-$IC$.
\end{proposition}

\begin{proof}
 Consider the following instance $I$ with 2 agents such that $a_1<a_2<\xi_1^-<d_2<\xi_1^+<d_1$ hold, where $\xi_1^-$ and $\xi_1^+$ denote the starting and ending times, respectively, of $\xi_1$\footnote{Such starting and ending times should exist since we have assumed that time intervals are contiguous, with the exception of $\xi_0$.}. Furthermore, the departure time $d_1$ belongs to interval $\xi_{0}$. It is easy to check that $\theta_1(I)=\{1\}$. Assume that agent 1 prefers $e_2$ to $e_1$, and agent 2 prefers $e_1$ to $e_2$. If agent 1 reveals her true departure time $d_1$, then Algorithm \ref{algoGreedyTTC} using the scheduled departure partition $\theta$ will assign her item $e_1$. On the other hand, if she misreports her departure time and reports $d'_1$ such that $\xi_1^-<d'_1<\xi_1^+$, then assuming that $I'$ denotes the copy of $I$ where agent 1 reveals $d'_1$ instead of $d_1$, it is easy to check that $\theta_1(I')=\{1,2\}$. Therefore, according to the preferences of agents 1 and 2, item $e_2$ is assigned by Algorithm \ref{algoGreedyTTC} to agent 1. Since $e_2$ is the most favorite item of agent 1, she has an incentive to misreport her departure time.
\end{proof}

\subsection{The Threshold-Triggered Partition}

We conclude this section by introducing a partition function that incentivizes agents to truthfully reveal both their arrival and departure times. First, note that if the scheduling consists of a single interval $\xi_0 = T$, then the scheduled departure partition $\theta$ corresponds to the partition in which each agent forms a singleton set. In this case, Algorithm~\ref{algoGreedyTTC} simply assigns each agent her own item. This algorithm is trivially $SIC$ since its outcome does not depend on the agents’ preferences or their arrival and departure times. However, it is not particularly interesting, as no exchanges are performed and the resulting allocation is not improved.

By combining the properties of Algorithms~\ref{algoParitionGamma} and~\ref{algoParitionTheta}, we construct a more meaningful partition function that still ensures Algorithm~\ref{algoGreedyTTC} is $SIC$. This function, denoted $\zeta$ and referred to as the \emph{threshold-triggered partition}, is defined by Algorithm~\ref{algoPartitionZeta}. The algorithm is equivalent to Algorithm~\ref{algoParitionTheta}, except that line~8 is replaced by the assignment $\theta \leftarrow \theta \cup \{\{\delta(i)\}\} \cup \{S \setminus \{\delta(i)\}\}$ (mirroring line~6 of Algorithm~\ref{algoParitionGamma}), and the input is restricted to a scheduling consisting of exactly two intervals, namely $\xi_0 = [t^-,\tau]$ and $\xi_1 = [\tau, t^+]$ for a given threshold $\tau \in \mathcal{T}$.\footnote{These two intervals are unambiguously determined by the threshold $\tau$, which therefore replaces in Algorithm \ref{algoPartitionZeta} the scheduling as the input.} Let $j$ denote the first agent whose departure time exceeds the threshold $\tau$. Then, the threshold-triggered partition $\zeta$ groups together all agents in $N_{< d_j}$, excluding agent~$j$ and any agents who departed before $\tau$. All other agents are placed in singleton sets. Finally, note that $\zeta$ is clearly $PP$, as it is derived from a variant of the algorithms defining the departing agent excluded partition $\theta$ and the scheduled departure partition $\gamma$.

\begin{algorithm}[tb]
\caption{Construction of the threshold-triggered partition $\zeta$.}
\textbf{Input}:  Instance $I=\{ (i, e_i, a_i, d_i, \succ_i)\}_{i\in N}$, and a threshold $t\in \mathcal{T}$.

\begin{algorithmic}[1]\label{algoPartitionZeta}
\STATE $thresholdPassed\leftarrow \mathtt{false}$. \COMMENT{Boolean that becomes \texttt{true} when $d_{\delta(i)}$ exceeds $t$.}
\STATE $B\leftarrow \emptyset$. \COMMENT{Set of agents already belonging to the partition.}
 \FOR{$i=1$ \textbf{to} $n$}
   \STATE \COMMENT{Iteration occurring at $d_{\delta(i)}$.}
   \IF{$\neg thresholdPassed$ and $d_{\delta(i)}\geq t$}
     \STATE $thresholdPassed\leftarrow \mathtt{true}$.\COMMENT{Threshold exceeded for the first time.}
     \STATE $\theta \leftarrow \theta \cup \{\{\delta(i)\}\} \cup \{ N_{<d_{\delta(i)}}\setminus (B\cup \{\delta(i)\})\}$.
     \STATE $B \leftarrow N_{<d_{\delta(i)}}$.
   \ELSE
     \STATE $\theta \leftarrow \theta \cup \{\{ \delta(i)\}\}$.
     \STATE $B \leftarrow B \cup \{ \delta(i)\}$.
   \ENDIF
 \ENDFOR
\end{algorithmic}
\end{algorithm}

\begin{example}
    Consider one last time instance $I$ described in the left part of Figure \ref{figOHMexample2}. Let us compute the threshold-triggered partition $\zeta(I)$ by running Algorithm \ref{algoPartitionZeta} with threshold $\tau=t^-$. Note that this means the first agent to leave after the threshold will be the first to exit the market. At $d_1$, during the earliest departure time of an agent, $N_{<d_1}$ contains agents 1, 2, and 3. Therefore, $\{1\}$ and $\{2, 3\}$ are added to the partition. All the other agents will remain alone in the partition. Thus, the algorithm halts and returns the threshold-triggered partition $\zeta(I) = \{\{1\}, \{2, 3\}, \{4\}, \{5\}\}$, which is identical to the departing agent excluded partition $\gamma(I)$. As a consequence, the matching returned by Algorithm \ref{algoGreedyTTC} using the threshold-triggered partition $\zeta$ is the matching $M$ depicted in Figure \ref{figOHMexample2}.

    Note that the outcome of the algorithm may vary depending on the threshold used. For example, consider a new threshold $\tau'$ such that $a_4 < \tau' < d_2$. In this case, the first agent to leave after the threshold is agent~2 at time $d_2$. The resulting partition is then $\zeta(I) = \{\{1\}, \{2\}, \{3, 4\}, \{5\}\}$. This partition is almost identical to the scheduled departure partition of Example~\ref{exSchedulePartition}, except that agents~1 and~2 do not belong to the same subset. However, since these two agents did not exchange their houses during the execution of Algorithm~\ref{algoGreedyTTC}, the outcome returned by Algorithm~\ref{algoGreedyTTC} for this partition $\zeta(I)$ remains $M'$.

\end{example}

\begin{proposition}\label{propZetaSIC}
   Algorithm \ref{algoGreedyTTC} using the threshold-triggered partition $\zeta$ is $SIC$.
\end{proposition}

\begin{proof}
    Note first that, since the threshold-triggered partition $\zeta$ is $PP$, we know from Proposition~\ref{propGreedyIC} that Algorithm~\ref{algoGreedyTTC}, when using $\zeta$ as input, is $WIC$. It remains to show that no agent has an incentive to misreport either her arrival or departure time. Let $S$ denote the only coalition in the threshold-triggered partition $\zeta$ that may contain more than one agent, and let $k$ be the agent whose departure time is the earliest among those occurring after the threshold~$\tau$. According to Algorithm~\ref{algoPartitionZeta}, we have $S = N_{< d_k} \setminus (B \cup \{k\})$, where $B$ is the set of agents whose departure times are earlier than~$\tau$. For any agent $j \notin S$, one of the following must hold:
\begin{enumerate}[label=(\roman*)]
    \item $j = k$;
    \item $a_j > d_k$;
    \item $d_j < \tau$.
\end{enumerate}

In cases~(i) and~(iii), agent~$j$ cannot alter her partition membership by misreporting her departure time $d_j$ as $d'_j$, since this would require $d'_j < d_j$. Indeed, this would imply $d'_j < d_j < \tau$ in case~(iii); and in case~(i), either $d'_k < \tau$, in which case agent~$k$ would leave before the threshold and be excluded from the coalition, or $d'_k \geq \tau$ and $d'_k$ would remain the earliest departure after~$\tau$, so the partition would remain unchanged. In case~(ii), agent~$j$ cannot manipulate her arrival time to join the coalition, since $a'_j > a_j$ would imply $a'_j > d_k$, and therefore her arrival would occur after agent~$k$ has departed. Finally, misreporting her arrival time in cases~(i) and~(iii), or her departure time in case~(ii), has no impact on her coalition, which remains a singleton. Thus, agents not in~$S$ have no incentive to misreport: they keep their own item under Algorithm~\ref{algoGreedyTTC}, regardless of their report.

Now consider an agent $j \in S$. The only way for her to alter the partition formed by the threshold-triggered partition $\zeta$ is to declare a departure time $d'_j < d_k$ (either $d'_j < \tau$, or $d'_j \in [\tau, d_k)$), or an arrival time $a'_j > d_k$. In both cases, agent~$j$ would be placed in a singleton set in the partition and would therefore keep her own item under Algorithm~\ref{algoGreedyTTC}. Since TTC is individually rational ($IR$), agent~$j$ has no incentive to misreport: she would receive an item at least as preferred as her own under truthful reporting. Hence, no agent has an incentive to misreport either her arrival or departure time, and the procedure is $SIC$.
\end{proof}

The main drawback of using the threshold-triggered partition $\zeta$ as input to Algorithm~\ref{algoGreedyTTC} is that only one coalition in $\zeta$ may contain more than one agent, and the exchanges performed during the TTC procedure may therefore involve very few agents. The number of agents participating in the TTC procedure depends heavily on the set of agents present in the market at the earliest departure time occurring after the threshold, and this number can range from $1$ to $n-1$.

\section{Conclusion and Future Works}

We extended the serial dictatorship and top trading cycle procedures to an online setting, aiming to develop mechanisms that are Pareto-efficient, individually rational, and incentivize agents to truthfully reveal their preferences, as well as their actual arrival and departure times. Several variants of these mechanisms were proposed and are summarized in Table \ref{tab:algorithm_comparison} along with their respective properties. It is important to note that Table \ref{tab:algorithm_comparison} alone may not suffice to fully compare different mechanisms. For instance, the static and dynamic versions of the serial dictatorship procedure using the ascending arrival permutation $\alpha$ (corresponding to rows 2 and 3 in Table \ref{tab:algorithm_comparison}) appear incomparable based solely on their listed properties. However, the allocation returned by the dynamic version is never Pareto-dominated by the allocation produced by the static version, while the reverse may hold, as demonstrated in Example \ref{exDynSerDict}. This observation suggests that the dynamic version is more efficient in practice than the static version. Another example involves the static top trading cycle procedure (Algorithm \ref{algoGreedyTTC}) using the threshold-triggered partition $\zeta$, which seems to dominate other partition functions since it satisfies a broader set of properties. However, $\zeta$ typically identifies only a single subset of agents where exchanges can occur. The size of this subset is highly dependent on the departure time of the earliest agent to leave after the threshold, making the mechanism less predictable and potentially inefficient in markets where agents frequently arrive and depart.

\begin{table}[ht]
\centering
\caption{Summary of the different mechanisms and their properties.}
\label{tab:algorithm_comparison}
\begin{tabularx}{\textwidth}{@{}ccccccc@{}} % Flexible column widths
\toprule
 Mechanism & $\mathcal{M}$-$PO$ & $\mathcal{S}$-$PO$ & $IR$ & $WIC$ & $a$-$IC$ & $d$-$IC$ \\
 \midrule
 {\small Algorithm \ref{algoPickingSequence} with ascending} & $\surd$ & & & $\surd$ & & \\
{\small departure permutation $\delta$}  & {\small (Prop. \ref{propParetoOpt})} & & & {\small (Prop. \ref{propOnlineWIC})} & &\\\midrule
 {\small Algorithm \ref{algoPickingSequence} with ascending} & & & & $\surd$ & & $\surd$ \\
 {\small arrival permutation $\alpha$} & & & & {\small (Prop. \ref{propOnlineWIC})} & & {\small (Prop. \ref{propAlphaDIC})}\\\midrule
 {\small Algorithm \ref{algoOnlinePickingSequence} with ascending} & & & & $\surd$ & & $\surd$ \\
 {\small arrival permutation $\alpha$} & & & & {\small (Prop. \ref{ICprop})} & & {\small (Prop. \ref{propdICalphadyn})}\\\midrule
 {\small Algorithm \ref{algoIROnlinePickingSequence} with ascending} & & $\surd$ & $\surd$ & & & \\
 {\small departure permutation $\delta$} & & {\small (Prop. \ref{propSafePO})} & & & &\\\midrule
 {\small Algorithm \ref{algoGreedyTTC} with departing} & & & $\surd$ & $\surd$ & & $\surd$ \\
  {\small agent excluded partition $\gamma$} & & & & {\small (Prop. \ref{propGreedyIC})} & & {\small (Prop. \ref{propGammaDIC})}\\\midrule
 {\small Algorithm \ref{algoGreedyTTC} with scheduled} & & & $\surd$ & $\surd$ & $\surd$ & \\
  {\small departure partition $\theta$} & & & & {\small (Prop. \ref{propGreedyIC})} & {\small (Prop. \ref{propThetaAIC})} &\\\midrule
 {\small Algorithm \ref{algoGreedyTTC} with threshold-} & & & $\surd$ & $\surd$ & $\surd$ & $\surd$ \\ 
  {\small triggered partition $\zeta$} & & & & {\small (Prop. \ref{propGreedyIC})} & {\small (Prop. \ref{propZetaSIC})} & {\small (Prop. \ref{propZetaSIC})}\\
\bottomrule
\end{tabularx}
\end{table}

The paper also presents additional results not summarized in Table \ref{tab:algorithm_comparison}. For example, we proved that the only mechanism guaranteeing a $\mathcal{M}(I)$-Pareto-optimal ($\mathcal{S}(I)$-Pareto-optimal, respectively) allocation for any instance $I$ is the static serial dictatorship procedure (safe serial dictatorship procedure, respectively) using the ascending departure permutation $\delta$ as input (Propositions \ref{propOnlyPO} and \ref{propSafePO}, respectively). As a consequence, $\mathcal{M}(I)$-$PO$ is incompatible with $IR$, $a$-$IC$, and $d$-$IC$ (Corollary \ref{corPOvsOthers}). Additionally, we demonstrated that the serial dictatorship procedure is inherently incompatible with $a$-$IC$, regardless of the input permutation function (Propositions \ref{propSerDictIncomAIC} and \ref{propDynSerDictIncompAIC}). Moreover, we established that the ascending departure permutation is the only input permutation function that ensures the serial dictatorship procedure satisfies $d$-$IC$ (Propositions \ref{alphaOnlyDIC} and \ref{propdICalphadyn}). Finally, we showed that no permutation function can make the safe serial dictatorship procedure $WIC$ (Proposition \ref{propSafeSerDictNoIC}).

An intriguing extension of this work would be to develop dynamic versions of the TTC procedure, allowing agents to participate multiple times during the execution of Algorithm \ref{algoGreedyTTC}, each time with their currently allocated item. Another promising avenue is to explore the design of strongly incentive-compatible mechanisms that facilitate more exchanges than Algorithm \ref{algoGreedyTTC} under the threshold-triggered partition $\zeta$. Additionally, relaxing our assumptions about the design of permutation and partition functions—specifically allowing them to depend on parameters beyond just the arrival and departure times of agents—opens another interesting research direction. Finally, investigating relaxed forms of strategy-proofness could provide valuable insights, particularly by accounting for agents' limited foresight regarding future arrivals of items and agents. Such an approach could incorporate probabilistic assumptions about arrival patterns, enabling mechanisms to achieve improved allocations. If the decision-maker has knowledge of these probabilities, optimal decisions could be crafted based on this information. These types of probabilistic assumptions are closely related to the online stochastic matching problem \cite{Feldman09,Huang21}. A similar approach has already been considered for the fair assignment of public goods \cite{banerjee2022online,Banerjee23}.

\vskip 0.2in
\bibliography{biblio}
\bibliographystyle{theapa}

\end{document}